\documentclass[a4paper, USenglish, cleveref, thm-restate]{lipics-v2021}

\nolinenumbers
\pdfoutput=1 

\EventEditors{Philip Bille, Seth Pettie, and Sabine Storandt}
\EventNoEds{3}
\EventLongTitle{34th Annual European Symposium on Algorithms (ESA 2026)}
\EventShortTitle{ESA 2026}
\EventAcronym{ESA}
\EventYear{2026}
\EventDate{August 31--September 4, 2026}
\EventLocation{L'Aquila, Italy}
\EventLogo{}
\SeriesVolume{388}
\ArticleNo{16}

\title{Efficient Uniform Negative Edge Weights}

\author{Lukas Geis}{Goethe University Frankfurt,
Germany}{lgeis@ae.cs.uni-frankfurt.de}{https://orcid.org/0009-0000-1687-2530}{}
\author{Daniel Allendorf}{Goethe University Frankfurt,
Germany}{dallendorf@ae.cs.uni-frankfurt.de}{https://orcid.org/0000-0002-0549-7576}{}
\author{Thomas Bläsius}{Karlsruhe Institute of Technology (KIT),
Germany}{thomas.blaesius@kit.edu}{https://orcid.org/0000-0003-2450-744X}{}
\author{Alexander Leonhardt}{Goethe University Frankfurt,
Germany}{aleonhardt@ae.cs.uni-frankfurt.de}{https://orcid.org/0009-0006-8263-6900}{}
\author{Ulrich Meyer}{Goethe University Frankfurt,
Germany}{umeyer@ae.cs.uni-frankfurt.de}{https://orcid.org/0000-0002-1197-3153}{}
\author{Manuel Penschuck}{University of Southern
Denmark}{penschuck@imada.sdu.dk}{https://orcid.org/0000-0003-2630-7548}{}
\author{Hung Tran}{Frankfurt Institute for Advanced Studies,
Germany}{htran@ae.cs.uni-frankfurt.de}{}{}

\authorrunning{L. Geis, D. Allendorf, T. Bläsius, A. Leonhardt, U.
Meyer, M. Penschuck, H. Tran}
\Copyright{Lukas Geis, Daniel Allendorf, Thomas Bläsius, Alexander
Leonhardt, Ulrich Meyer, Manuel Penschuck, and Hung Tran}

\ccsdesc[500]{Mathematics of computing~Graph algorithms}

\keywords{Random Graphs, Shortest Path, Random Edge Weights, Negative Cycles}

\acknowledgements{In memory of Daniel,  whose sharp analytical mind
  and keen insights profoundly influenced this work. \\\ \\
  This work was partially funded by the Deutsche
  Forschungsgemeinschaft (DFG) --- ME 2088/5-1, ME-2088/5-2 (FOR 2975
  --- Algorithms, Dynamics, and Information Flow in Networks), the
  LOEWE program of the state Hesse (CMMS), and the Novo Nordisk
  Foundation grant NNF21OC0066551.
  We thank Deepak Ajwani and Ryan O'Conner for valuable discussions
  on shortest path algorithms with negative weights that triggered
  this research, and Mia Ischebeck for her preliminary investigation
  of a related MCMC process.
}

\supplementdetails{Generator source and analysis scripts /
  Experimental Data (Uncompressed 35GB / Compressed
12GB)}{https://doi.org/10.5281/zenodo.21128649}

\usepackage{etoolbox}
\usepackage{intcalc}
\usepackage{tikz}
\usepackage[nocompress]{cite}

\usepackage{caption}
\usepackage{subcaption}
\usepackage{multirow}
\usepackage{todonotes}
\usepackage{mathtools}

\usepackage[vlined,ruled,linesnumbered,algo2e]{algorithm2e}
\DontPrintSemicolon
\SetKwInput{KwInput}{Input}
\SetKwInput{KwOutput}{Output}

\def\weightUp#1#2#3{\ensuremath{#1_{#2 \gets #3}}\xspace}
\def\setCons{\ensuremath{\mathbb S}\xspace}
\def\setConsFull{\ensuremath{\setCons_{G, \wInt}}\xspace}
\def\probCons{\ensuremath{\mathcal S}\xspace}
\def\algdk{\textsc{Dijkstra}\xspace}
\def\algbd{\textsc{BiDijkstra}\xspace}
\def\algbf{\textsc{BellmanFord}\xspace}
\def\algjohn{\textsc{Johnson}\xspace}
\def\wmax{\ensuremath{w_\text{max}}\xspace}
\def\wunif{\ensuremath{w_\text{unif}}\xspace}
\def\wone{\ensuremath{w_\text{one}}\xspace}
\def\wzero{\ensuremath{w_\text{zero}}\xspace}
\def\negbroken{\ensuremath{w'_{\phi}(u,v)}}
\def\broken{\ensuremath{-w'_{\phi}(u,v)}}
\def\ie{i.\,e.,\xspace}
\def\eg{e.\,g.\xspace}
\def\cf{c.\,f.\xspace}
\def\etal{et\,al.\xspace}
\def\gnp{\ensuremath{\mathcal{GNP}}\xspace}
\def\rhg{\ensuremath{\mathcal{RHG}}\xspace}

\def\roads{\ensuremath{\mathcal{ROAD}}\xspace}
\def\Oh{\ensuremath{\mathcal{O}}}
\def\avgdeg{\ensuremath{\overline{d}}\xspace}
\def\nodepath#1#2{\ensuremath{#1{\rightarrow}#2}\xspace}
\def\mcsteps{\ensuremath{\tau}\xspace}
\def\set#1{\ensuremath{\left\{#1\right\}}}
\def\setc#1#2{\set{#1 \mid #2}}

\def\wInt{\ensuremath{\mathcal W}}
\def\rej{\textsc{RejSampler}\xspace}
\def\norej{\textsc{NoRejSampler}\xspace}
\def\mincyc{\textsc{MinCycle}\xspace}
\def\resamp{\textsc{Resamp}\xspace}

\DeclareMathOperator{\Gap}{Gap}

\usepackage{enumitem}

\usepackage{siunitx}
\usetikzlibrary{decorations.pathreplacing,calligraphy,calc}

\let\epsilon\varepsilon

\begin{document}

\clearpage

\maketitle

\begin{abstract}
  We consider a maximum entropy edge weight model that allows for
  negative weights.
  Given a graph~$G$ and possible weights $\wInt$ typically consisting
  of positive and negative values, the model selects edge weights ~$w
  \in \wInt^m$ uniformly at random from all weights that do not
  introduce a negative cycle.
  We propose an MCMC process and show that it converges to the
  required distribution.
  We then engineer an implementation of the process using a dynamic
  version of \algjohn's algorithm in connection with a bidirectional
  \algdk search as well as an innovative resampling method.
  We empirically study the performance characteristics of these novel
  sampling algorithms as well as the output produced by the model.
\end{abstract}

\smallskip

\newpage
\section{Introduction}\label{sec:intro}
Shortest path problems are a fundamental algorithmic challenge with
applications in network routing, geographical navigation, and other
optimization tasks.
Classic solutions such as \algdk's
algorithm~\cite{DBLP:journals/nm/Dijkstra59} and most research into
practical routing tools operate under the constraint of non-negative
edge weights.

Yet, real-world scenarios, such as the routing of recuperating
electric vehicles with energy gains and expenditures, naturally
present graphs with negative weights;
other shortest path problems with negative weights arise, for
instance, in the context of flow problems (\eg \cite{10.1145/321694.321699}).
For meaningful shortest paths to exist, graphs should not contain
cycles of negative length, as otherwise, paths of unbounded negative
weight arise.

Analytical and empirical studies of shortest path algorithms
frequently involve random graph models (\eg
\cite{DBLP:journals/dam/FriezeG85,DBLP:journals/jea/BorassiN19}) as
they tend to have a well-understood and mathematically tractable
structure~\cite{Barabasi2016-np,DBLP:books/cu/H2016}.
Typically, the graph topology and edge weights are considered separately.
While considerable efforts were put towards ``realistic'' random
graphs~(\eg \cite{DBLP:journals/csur/DrobyshevskiyT20} for a survey),
maximum entropy network models are frequently used for hypothesis
testing, baselines, and null models~\cite{Milo824, gotelli1996null, Peixto15}.
Applications in biology~\cite{DEMARTINO2018e00596} and reinforcement
learning~\cite{DBLP:journals/make/MoosHASCP22} indicate that in
practice they can improve inference quality in different domains,
either by increasing the models' robustness or by decreasing their bias.

To maximize the entropy of random edge weights, we consider weights
sampled independently and uniformly from some set $\wInt$; typically,
the interval $\wInt = [a,b]$ is a plausible choice.
In the non-negative case ($a \ge 0$), a straightforward sampling
method is to process each edge in isolation.
This does not hold if a significant probability mass goes towards
negative weights ($a < 0$ and $|a| / (b-a) \gg 0$), since cycles of
negative weights are likely to appear in many graphs.
For example, consider a doubly-linked path, \ie a graph of nodes
$v_1, \ldots, v_{k+1}$ with edges $(v_i, v_{i+1})$ and $(v_{i+1},
v_i)$ for all $1 \le i \le k$.
Assuming random i.i.d.\ edge weights with $\wInt = \set{-1, 0, 1}$,
each pair is randomly assigned one out of $|\wInt^2| = 9$
possibilities; exactly six of these choices are legal as they have
non-negative sums.
Hence, the probability of no negative cycles vanishes with an upper
bound of $(6/9)^{k}$.

\subsection{Our contribution}
In \cref{sec:model}, we extend the maximum entropy model for signed
edge weights without introducing negative cycles.
We then propose a basic Markov chain process to sample random
negative edge weights in \cref{sec:sampler} and show that the process
converges to a uniform distribution.
In \cref{sec:engineering}, we discuss multiple algorithmic
optimizations to accelerate the sampling process.
These involve a simple-to-implement dynamic variant of \algjohn's
algorithm including an application of a bidirectional SSSP search.
We then also show how to further bias the process to force a faster
convergence using a novel resampling method without losing uniformity.
Lastly, in \cref{sec:experiments}, we empirically study the
performance and convergence of the MCMC process as well as the
properties of the output produced.
This yields a highly optimized sampling algorithm, which we provide
ready-to-use on~\url{https://codeberg.org/lukasgeis/unew}.

\subsection{Preliminaries and notation}
Let $G = (V, E)$ be a \emph{directed graph} with $E \subseteq V \times V$.
If unambiguous from context, we use $n$ and $m$ to denote the number
of nodes~$|V|$ and edges~$|E|$, respectively.
Given a graph~$G = (V,E)$, we call $w\colon E \rightarrow \mathbb{R}$
an \emph{edge weight} function (short weights);
by fixing some implicit order on $E$, we represent the function as a
vector $w \in \mathbb{R}^m$.
We denote new weights, that differ from~$w$ only in one position
corresponding to a single edge~$e$, as $\weightUp{w}{e}{c}$.

We call $P = (v_1, \ldots v_k) \in V^k$ a $k$-path, if all edges
$(v_i, v_{i+1})$ exist and ``overload'' the weight function for paths
as $w(P) = \sum_i w((v_i, v_{i+1}))$.
A path~$P$ is \emph{negative} if $w(P) < 0$.
A $(k+1)$-path is a \emph{$k$-cycle} if $v_{k+1}  = v_{1}$.

For a graph~$G$ with weights~$w$, we say \emph{$(G, w)$ is
consistent} if no negative cycles exist.
Two vertices $u, v \in V$ are \emph{strongly connected} if there
exists a directed \nodepath{u}{v} path and vice versa.
Being strongly connected is an equivalence relation on $V$ and the
equivalence classes are called \emph{strongly connected components} (SCC).
A graph is \emph{strongly connected} if it has only one strongly
connected component.

\subsection{Related work}\label{subsec:related_work}
Sampling negative edge weights was studied before (\eg
\cite{DBLP:journals/mp/CherkasskyG99,DBLP:conf/wae/DemetrescuFMN00}),
typically as a secondary consideration to synthesize benchmark data sets.
Yet, we are not aware of any existing maximum entropy models for
negative edge weights in shortest path networks.
However, given the conceptual simplicity of our sampling algorithm,
similarities with existing work are expected.
Notably, \cite{DBLP:conf/wae/DemetrescuFMN00} empirically study
shortest path under edge weight updates and, among others, use their
\texttt{gen\_seq} algorithm (described in five lines of prose) to
produce update sequences.
We expect that such sequences behave similarly to the transitions in
our Markov chain proposed in~\cref{sec:sampler}.
However, given the vastly different scope of their work, the authors
did not analyse the process nor the properties of the output.

As we are unaware of previous work on maximum entropy models for
negative edge weights, we mainly review the practical need for such
networks and shortest path algorithm. We do, however, note that there
is extensive research into Monte-Carlo-Markov-Chains (MCMC) for
approximate uniform samplers of different graph problems (\eg
\cite{rao1996markov,DBLP:conf/alenex/GkantsidisMMZ03,strona2014curveball,CarstensPhd}).

A clear strength of such MCMC solutions are their simplicity and
relative ease with which additional constraints can be added (\eg
\cite{DBLP:conf/alenex/StantonP11,DBLP:journals/compnet/VigerL16,DBLP:conf/dexa/ArafatBDB20}).
Notably, \cite{DBLP:journals/compnet/VigerL16} augment Edge Switching
with graph connectivity tests.
This bears some resemblance to our proposal as both processes
frequently execute reachability/shortest path computations.
However, aside from addressing different problems, the methods also
differ in technical details.

\subparagraph*{Networks with negative edge weights}
A prominent application are road networks for electrical cars where
downward-directed roads allow the vehicles to recuperate their batteries.
A recent survey summarized a plethora of applications and research
directed towards negative edge weights in social networks~\cite{KAUR201621}.
Here, positive edge weights might correspond to friendship or
affection while negative edge weights relate to opposition and aversion.
This yields new measures of power in politically charged
networks~\cite{DBLP:journals/socnet/SmithHKLBB14} and centrality
measures respecting negative ties~\cite{DBLP:journals/socnet/EverettB14}.

Several ``nature inspired'' techniques produce shortest path networks
with negative weights.
For instance, a common technique derives edge weights from the
gradient of random node potentials.
This, however, needs further refinements (\eg by adding positive
random noise) since otherwise, each cycle has weight exactly~$0$
which seems implausible for physical networks.

\subparagraph*{Shortest path algorithms}
The SSSP algorithms \algdk~\cite{DBLP:journals/nm/Dijkstra59}, and
\algbf~\cite{Bel58,Ford56} are classical, but still relevant solutions.
With a runtime of $\Oh(m+n\log(n))$, \algdk is asymptotically faster
than the $\mathcal{O}(nm)$ runtime of \algbf.
Still, \algdk has the significant drawback of only supporting
non-negative edge weights.
\algjohn's all-pair shortest path algorithm implements an ingenious
workaround~\cite{DBLP:journals/jacm/Johnson77}.
It translates general edge weights solely to non-negative weights,
such that a shortest path in one network translates to a shortest
path in the other (see \cref{lem:johnson} for details).

In addition, the recent theoretical breakthrough result of
Bernstein~\etal~\cite{DBLP:conf/focs/BernsteinNW22} presents an
algorithm that solves the SSSP problem on networks with negative edge
weights but without negative cycles in near-linear time.
This recent addition stimulated subsequent (partial) implementations
thereof~\cite{DBLP:conf/wea/CassisKNR25} which had to resort to
benchmarking their algorithmic engineering efforts on (biased) weight
distributions in lack of a maximum-entropy generator.

\section{The Negative Edge Weight Model}\label{sec:model}
We want to uniformly sample from the set of all edge weights
consistent with~$G$:
\begin{definition}
  Given graph~$G$ and weights~$\wInt$, define the set of consistent
  edge weights as
  $$\setConsFull = \setc{w}{w\colon E\mapsto \wInt\ \land\ (G,
  w)\text{ is consistent}}.$$
  Then, $\probCons_{G, \wInt}$ is the uniform distribution over
  $\setCons_{G, \wInt}$.
  If unambiguous, we shorten to $\setCons$ and $\probCons$.
\end{definition}

\subparagraph*{Assumptions}
For ease of exposition, we assume without loss of generality that
(A1) $G$ is strongly connected and (A2) $\wInt$ contains positive and
negative values.
If (A1) does not hold, we may decompose~$G$ into strongly connected
components and a residual acyclic graph in linear
time~\cite{DBLP:journals/siamcomp/Tarjan72}, process the former
independently, and trivially sample weights for the latter.
As for (A2), observe that sampling from $\probCons_{G,\wInt}$ is
trivial if $G$ is acyclic or if $\wInt$ does only contain positive
values, as neither case can have negative cycles.
Further, if $\wInt$ consists only of negative values, consistent
weights exist iff $G$ is acyclic.

\subparagraph*{Variants}
It is natural to ask whether ---for some number~$k \in \mathbb N$---
there exists a $w \in \setCons$ that has $k$ negative edges.
Unfortunately, this problem is NP-hard.
This can be seen by a simple reduction from the NP-hard problem
\emph{Directed Feedback Arc Set}~\cite{DBLP:conf/coco/Karp72}, which
asks whether one can remove $k$ edges from a directed graph to obtain
an acyclic graph.
This is the case iff there exists a consistent weight function $w$
assigning negative weights to $m - k$ edges.
By definition, the graph induced by the negative edges is acyclic;
removing the $k$ positive edges thus suffices to obtain an acyclic graph.
Conversely, assume that removing $k$ edges results in an acyclic graph.
Then assigning large positive weights to these $k$ edges and small
negative weights to the remaining edges results in a consistent weight function.

\section{The MCMC sampler}\label{sec:sampler}
\begin{algorithm2e}[tb]
  \caption{\rej: Approximate uniform sampler}
  \label{algo:sampler}

  \small

  \KwInput{graph $G = (V,E)$, weight interval~\wInt, number of steps \mcsteps}
  \KwOutput{consistent weights $w\colon E \rightarrow \wInt$}

  \BlankLine
  $w \leftarrow $ arbitrary consistent weight function\;
  \For{$t \gets 1 \text{ to } \mcsteps$}{
    $e \leftarrow$ uniform random edge in $E$\;
    $c \leftarrow$ uniform random weight in $\wInt$\;
    \If{$(G,\weightUp{w}{e}{c})$ is consistent \label{line:sampler_consistent}}{
      accept $w \leftarrow \weightUp{w}{e}{c}$ \tcp*{update weight of~$e$}
    }
  }
  \KwRet{$w$}\;
\end{algorithm2e}

Unfortunately, a direct sampling method of $\probCons$ is highly
non-trivial as each cycle imposes a constraint on the weight function.
Since there are possibly an exponential number of cycles in a graph
and only an exponentially small fraction of all weights in $\wInt^m$
appear in $\setCons$, we instead opt for an approximate sampler using
a Markov-Chain Monte-Carlo (MCMC) process.
We introduce \cref{algo:sampler} to sample approximately uniformly
from $\setCons_{G, \wInt}$ for a strongly connected graph~$G$ and
partially negative values~$\wInt$ (wlog; see \cref{sec:model}).

We start with arbitrary weights~$w \in \setCons_{G, \wInt}$;
canonical choices include setting all edge weights to
$0$ (if $0 \in \wInt$),
$1$ (if $1 \in \wInt$),
the largest value in $\wInt$,
or a randomly chosen non-negative value in $\wInt$.
Subsequently, we perturb the solution by repeating the following step
$\mcsteps$ times:
we select a uniform random edge $e \in E$ and new weight $c \in \wInt$.
Then, if $(G, \weightUp{w}{e}{c})$ remains consistent, we update the
edge and reject otherwise.

For simplicity, we assume in the following that the weight interval
$\wInt$ is discrete and finite, thus rendering the state space finite.
The spirit of all arguments and results, however, carries over to
continuous~$\wInt$.
The state space of the MCMC in \cref{algo:sampler} consists of
exactly all consistent weight functions in~$\setCons$.
The neighborhood~$N[w]$ of state $w$ consists of all weight functions
that differ in at most one position:

$$
N[w] = \setc{w'}{\exists e \in E\ \exists c \in \wInt\colon w' =
\weightUp{w}{e}{c} \in \setCons}
$$

\begin{observation}[Symmetry]\label{obs:symmetric-edges}
  Weights $w_1, w_2 \in \setCons$ satisfy $w_1 \in N[w_2]
  \Leftrightarrow w_2 \in N[w_1]$.
\end{observation}

\begin{observation}\label{obs:may-increase}
  Let $w_1 \in \setCons$ and $w_2 \in \wInt^m$, s.t.
  $\forall e \in E\colon \ w_1(e) \le w_2(e)$.
  Then, $w_2 \in \setCons$.
  We can reach $w_2$ from $w_1$ by increasing the weight in all
  differing positions in any order.
\end{observation}

\begin{theorem}\label{thm:convergence-mc}
  Let $G = (V, E)$ be a directed graph and $\wInt$ a weight interval.
  Then, the MCMC process converges to a uniform distribution on
  $\setCons_{G, \wInt}$.\label{thm:uniform}
\end{theorem}

\begin{proof}
  We use the fact that any irreducible,  aperiodic, and symmetric
  Markov chain converges to a uniform distribution
  \cite[Th.~7.10]{DBLP:books/daglib/0012859}.
  We only show these properties:

  The MC is \underline{irreducible} iff its state graph is strongly connected.
  Let $w_1, w_2 \in \setCons$ be states and $\wmax$ the function that
  assigns every edge the maximum weight $\max(\wInt)$.
  By \cref{obs:may-increase}, we know that there exist paths $P_i$
  from $w_i$ to \wmax.
  Then, by \cref{obs:symmetric-edges}, the opposite paths $P^r_i$
  from \wmax to $w_i$ also exist.
  We can therefore reach $w_2$ from $w_1$ by $P_1$ followed by
  $P^r_2$ (and vice versa). \hfill $\bigtriangleup$

  An irreducible MC is \underline{aperiodic} if there exists a self-loop.
  In fact, each state~$w$ has a loop:
  If we sample edge $e$ and the old weight $c = w(e)$, no transition
  occurs. \hfill $\bigtriangleup$

  The MC is \underline{symmetric} iff the probability of transition
  between any two adjacent states is identical in both directions.
  Let $w_1, w_2 \in \setCons,w_1\in N[w_2]$.
  If $w_1 = w_2$, the property holds trivially.
  For $w_1 \ne w_2$, there is a unique $e \in E$ with $w_1(e) \ne w_2(e)$.
  To transition between $w_1$ and $w_2$, we need to draw $e$ with
  probability $1/m$ and $w_1(e)$ or $w_2(e)$ with $1/|\wInt|$.
\end{proof}

We show in the appendix that (i) \cref{thm:convergence-mc} also holds
true for uncountable $\wInt$ and (ii) the process is rapidly mixing
on the simple $n$-cycle graph.

\section{Algorithmic insights}\label{sec:engineering}
Most of \cref{algo:sampler} can be efficiently implemented using
standard techniques.
We maintain the graph in the \emph{compressed sparse row} (CSR)
format~\cite{tinney1967direct}, allowing fast neighborhood iteration
as well as constant time uniform sampling of an edge.

The most expensive task of \cref{algo:sampler} is to test whether the
updated weight function~$w'$ remains consistent with $G$.
To this end, we need to verify that $w'$ induces no negative cycle.
As discussed in \cref{subsec:related_work}, this problem has already
been considered for several applications.
However, we are unaware of solutions that utilize bidirectional
searches, which often perform significantly better as shown in
\cref{sec:experiments}.
In the following, we begin with a straightforward unidirectional
implementation designed for the MCMC process and highlight
connections to previous work.
Subsequently, we extend the algorithm to incorporate a bidirectional search.
Finally, we alter the algorithm to always sample an acceptable weight
uniformly for an edge.

Consider a weight update $w' = \weightUp{w}{e}{c}$ on edge~$e$.
In the following, we focus on decreased weights (\ie $w'(e) < w(e)$),
since increased weights cannot introduce negative cycles by
\cref{obs:may-increase}.
\algbf~\cite{Ford56,Bel58} can detect negative cycles in~$w'$ directly.
However, even an optimized implementation performs subpar.
The primary issue is \algbf's lack of locality --- the algorithm
typically performs several global rounds despite the fact that
negative cycles tend to be short.

\medskip

Thus, we want to switch to \algdk and later its bidirectional
variant; neither of which supports negative weights.
As a first step, we show how to detect negative cycles in $w'$ using
only the previous weights~$w$.
The main insight is inductive in nature: as the current $(G, w)$ is
consistent (otherwise they would have been rejected), we know that
the candidate $(G, w')$ is consistent iff the updated edge~$e$ is not
part of a negative cycle.
This can be tested using:

\begin{observation}\label{obs:new_weight_in_cycle}
  Let $P$ be a shortest \nodepath{v}{u} path in $(G, w)$.
  Obtain $w' = \weightUp{w}{e}{c}$ by updating edge~$e=(u,v)$.
  Then, $e$ is in a negative cycle for $w'$ iff $w'(u,v) + w(P) < 0$.
\end{observation}

\subsection{An online version of Johnson's algorithm}
If none but the new weight~$w'(e)$ was negative, we could use \algdk
based on \cref{obs:new_weight_in_cycle} with a running time of
$\mathcal{O}(m + n\log n)$.
Even better, we can heavily prune the search space and ignore all
paths that have a weight of at least $-w'(e)$:
Since all unprocessed edges have a non-negative weight, they cannot
lead to strictly lighter paths than the ones already discovered (\ie
the smallest element in \algdk's priority queue).

However, since $w$ may contain negative weights, we recast the
shortest path problem into an equivalent non-negative analog using
the potential
method~\cite{DBLP:journals/jacm/Johnson77,DBLP:conf/focs/BernsteinNW22}.
The following lemma provides the core ingredient for this to work and
also forms the basis for our algorithm.  The function $\phi$ in the
lemma statement is also called \emph{potential function}.
\begin{lemma}[based on~\cite{DBLP:journals/jacm/Johnson77}]\label{lem:johnson}
  Let $G = (V, E)$ be a graph with weights $w\colon E \to \mathbb R$
  and let $\phi\colon V \to \mathbb R$.
  Further, let $w_\phi(u, v) = w(u, v) + \phi(v) - \phi(u)$.
  Then, any path in $G$ is a shortest path with respect to $w$ if and
  only if it is a shortest path with respect to $w_\phi$.
\end{lemma}
\begin{proof}
  Let $P = (v_1,\ldots,v_p)$ be a path in $G$.
  The path weight~$w_{\phi}(P)$ follows as a telescope sum that
  cancels out all contributions but the first and the last potential:
  \begin{align}
    w_{\phi}(P) = & \sum\nolimits_{i=1}^{p-1} \left[
      \textcolor{red}{w(v_i,v_{i+1})} - \textcolor{blue}{\phi(v_i)} +
    \textcolor{green!50!black}{\phi(v_{i+1})} \right] \nonumber \\
    & = \textcolor{red}{\sum_{i=1}^{p-1} w(v_i,v_{i+1})}
    - \textcolor{blue}{\sum_{i=1}^{p-1} \phi(v_i)}
    + \textcolor{green!50!black}{\sum_{i=2}^{p} \phi(v_i)}
    \label{eq:telescope_sum}
    = \textcolor{red}{w(P)} - \textcolor{blue}{\phi(v_1)} +
    \textcolor{green!50!black}{\phi(v_p)}
  \end{align}
  The remaining potentials $\phi(v_1)$ and $\phi(v_p)$ appear on all
  \nodepath{v_1}{v_p} paths.
  Thus, a shortest path for $w_\phi$ remains a shortest path in terms
  of $w$ (and vice versa).\hfill
\end{proof}

\noindent
We call a potential function $\phi$ \emph{feasible} if $w_\phi(u, v)
\ge 0\  \forall (u, v) \in E$.
Further, we call $w_{\phi}(u,v)$ the \emph{potential weight} of $(u,v)$.
We say edge $(u,v)$ is \emph{broken} if it has a negative
\emph{potential weight}.

\subsection{Maintaining a feasible potential}
\begin{figure}[t]
  {
\def\dx{4em}
\def\dy{2.2em}

\def\edge#1#2#3{
    \edef\pota{\ifcsdef{pot#1}{\csuse{pot#1}}{\csuse{orgpot#1}}}
    \edef\potb{\ifcsdef{pot#2}{\csuse{pot#2}}{\csuse{orgpot#2}}}
    \numdef\wphi{\csuse{w#1#2} + \potb - \pota}
    \edef\broken{\ifnumcomp{\wphi}{<}{0}{broken}{edge}}
    \path[edge, \broken, #3] (v#1) to node[align=center, sloped] {$\csuse{w#1#2}$ \\ $\wphi$} (v#2);
}

\tikzset{
    vertex/.style = {fill=white, draw, inner sep=0, minimum width=1.5em, minimum height=1.5em, circle},
    normal/.style={vertex},
    high/.style={vertex, fill=white!90!red, draw=red},
    org vertex/.style={vertex, green!50!black, fill=white, opacity=0.5},
    edge/.style={draw, ->},
    broken/.style={edge, thick, red}
}

\def\commonElements#1#2{
\foreach \y in {0, ..., #1} {
        \draw[dotted, opacity=0.5] (-0.5 * \dx, \y * \dy) to node[left, pos=0] {\y} ++(4.75 * \dx, 0);
        \draw (-0.5 * \dx, \y * \dy) to node[left, pos=0] {\y} ++(0.125 * \dx, 0);
    }

\draw[->] (-0.375*\dx, -0.125*\dy) to node[left, pos=1] {$\phi$} ++(0, 0.5*\dy + #1 * \dy);

\foreach \name/\prop [count=\x] in {z/vertex, u/high, v/high, x/vertex, y/vertex} {
\edef\pot{\ifcsdef{pot\name}{\csuse{pot\name}}{\csuse{orgpot\name}}}
\ifnumcomp{\pot}{=}{\csuse{orgpot\name}}{
\edef\opa{\ifcsdef{pot\name}{0.4}{1}}
\expandafter\node\expandafter[\prop, opacity=\opa, label=above:{\ifthenelse{\equal{#2}{nolabel}}{}{\tiny $\phi(\name){=}\pot$}}]
(v\name) at (\dx * \x - \dx, \dy * \pot) {$\name$};
}{
\expandafter\node\expandafter[\prop, label=above:{\ifthenelse{\equal{#2}{nolabel}}{}{\tiny $\phi(\name){=}\pot$}}]
(v\name) at (\dx * \x - \dx, \dy * \pot) {$\name$};
\node[org vertex] (org-v\name) at (\dx * \x - \dx, \dy * \csuse{orgpot\name}) {$\name$};
\draw[thick, ->, green!50!black, opacity=0.5] (org-v\name) to (v\name);
}
}

\edge{u}{v}{};
\edge{v}{x}{};
\edge{x}{y}{};
\edge{z}{u}{};
}

\def\orgpotz{1}
\def\orgpotu{1}
\def\orgpotv{1}
\def\orgpotx{0}
\def\orgpoty{0}

\begin{center}
    \scalebox{0.8}{
        \begin{tikzpicture}[
                header label/.style={anchor=east},
                weight plot/.style={xshift=15em}
            ]
            \node[header label] (inp) at (0,0) {
                \begin{minipage}{16em}\textbf{Input:} \\
                    Since nodes $u$ and $v$ have same potential,
                    the edge $(u,v)$ is broken: $w_\phi(u,v) = w(u,v) = -2$.
                \end{minipage}
            };
            \node[weight plot] at (inp.east) {
                \begin{tikzpicture}
                    \def\wuv{-2}
                    \def\wvx{ 1}
                    \def\wxy{ 1}
                    \def\wyz{-1}
                    \def\wzu{ 1}

                    \commonElements{2}{}
                    \edge{y}{z}{bend left, out=310, in=230};
                \end{tikzpicture}};

            \node[header label] (sol1) at (0, -3) {
                \begin{minipage}{16em}\textbf{Solution 1 (using \algdk):} \\
                    increase only $\phi(v) + 2$.
                    This large update causes ``many'' cascading updates
                \end{minipage}};
            \node[weight plot] at (sol1.east) {\begin{tikzpicture}
                    \def\potz{2}
                    \def\potu{1}
                    \def\potv{3}
                    \def\potx{2}
                    \def\poty{1}

                    \def\wuv{-2}
                    \def\wvx{ 1}
                    \def\wxy{ 1}
                    \def\wyz{-1}
                    \def\wzu{ 1}

                    \commonElements{3}{nolabel}
                    \edge{y}{z}{bend right, out=40, in=130};
                \end{tikzpicture}};

            \node[header label] (sol2) at (0, -6.5) {\begin{minipage}{16em}
                    \textbf{Solution 2 (using \algbd):}\\
                    split update to $\phi(u) - 1$ and $\phi(v) + 1$.
                    This update causes fewer cascading updates as $z$ and $y$ remain unchanged.
                \end{minipage}
            };
            \node[weight plot] at (sol2.east) {\scalebox{0.95}{\begin{tikzpicture}
                        \def\potz{1}
                        \def\potu{0}
                        \def\potv{2}
                        \def\potx{1}
                        \def\poty{0}

                        \def\wuv{-2}
                        \def\wvx{ 1}
                        \def\wxy{ 1}
                        \def\wyz{-1}
                        \def\wzu{ 1}

                        \commonElements{2}{nolabel}
                        \edge{y}{z}{bend right, out=25, in=130, opacity=0.4};

                        \draw [decorate, decoration = {calligraphic brace}, thick] (0.6*\dx, 2) to node[above] {backward} ++(0.8*\dx, 0);
                        \draw [decorate, decoration = {calligraphic brace}, thick] (1.6*\dx, 2) to node[above] {forward search} ++(1.8*\dx, 0);
                    \end{tikzpicture}}};
        \end{tikzpicture}
    }
\end{center}
}
  \vspace{-1.5em}
  \caption{
    Fixing potential on a cycle graph.
    The number above edge~$(a,b)$ is $w(a,b)$ and below
    $w_{\phi}(a,b) = w(a,b) + \phi(b) - \phi(a)$;
    green nodes indicate the input potential (top panel).
    Edge $(u,v)$ is broken in the input. It needs an increase of
    $\phi(v) - \phi(u)$ by at least $-w_\phi(u,v) = 2$ to make $\phi$ feasible.
  }
  \label{fig:potentials}
\end{figure}

A feasible potential function~$\phi$ can be computed by adding a
node~$r$ with outgoing edges of weight~$0$  to all existing nodes.
Then, SSSP starting in~$r$ yields a feasible $\phi$ by setting
$\phi(u)$ to minus the length of the shortest $\nodepath r u$ path
found~\cite{DBLP:journals/jacm/Johnson77}.
We avoid this step by starting the MCMC process with non-negative edge weights.
Then all potentials can be set to~$0$.

We thus focus on the dynamic problems of (i) testing for negative
cycles and (ii) maintaining feasible potentials after edge weight updates.
The following lemma allows us to answer (i) as a direct byproduct of
(ii) and is similar to \cite[Thm. 3]{10.1145/321694.321699}:

\begin{lemma}\label{lem:infeasible_implies_negative_cycle}
  If $\phi$ is a feasible potential function for $G$ and $w$, $(G,w)$
  is consistent.
\end{lemma}
\begin{proof}
  Consider any cycle $C$ in $G$ and observe that we can interpret $C$
  as a path with $v_1 = v_p$.
  Then, \Cref{eq:telescope_sum} implies that the weight of~$C$ is
  independent of the potential function, \ie
  $w_{\phi}(C) = w(C) + \phi(v_p) - \phi(v_1) = w(C).$
  Then, if $\phi$ is feasible, every edge in $G$ with respect to
  $w_{\phi}$ has a non-negative weight.
  Thus, every cycle must also have a non-negative total weight.
  But since $w_{\phi}(C) = w(C)$ for every cycle $C$, there cannot
  exist a negative weight cycle in $G$ if $\phi$ is feasible.
  \hfill
\end{proof}

Thus, our optimized implementation maintains a potential function
$\phi$ and checks whether an update is consistent by attempting to
maintain a feasible~$\phi$.
Further, we only update~$\phi$ under weight decreases since a weight
increase only increases the potential weight of the edge and thus
cannot break \emph{feasibility} of $\phi$.

However, as illustrated in \cref{fig:potentials}, decreasing an edge
weight $w(u,v)$ may \emph{break} the edge.
To make $\phi$ feasible again, we need to increase the gradient
$\phi(v) - \phi(u)$.
This can be achieved by (i) decreasing $\phi(u)$, (ii) increasing
$\phi(v)$, or (iii) both.
We begin with variant (ii) and discuss option~(iii) in
\cref{subsec:bidirectional}.
In any case, changing one potential may cause ``collateral damage''
by breaking additional edges.
The following lemma shows how to efficiently repair cascading
breakages using a single (partial) shortest path query.

\begin{lemma}\label{lem:potential_update}
  Consider weights~$w'$ obtained from $w$ by decreasing the weight of $(u, v)$.
  If $w'$ has no negative cycles, we can construct a feasible $\phi'$
  with $\phi'(x) = \phi(x) + \max\{0, -w'_{\phi}(u,v) - D(x)\}$ where
  $D(x)$ is the length of a shortest \nodepath{v}{x} path with
  respect to $w_{\phi}$.
\end{lemma}
\begin{proof}
  Note that if $\negbroken > 0$, we have $\max\{0, \broken - D(x)\} =
  0\ \forall x\in V$ since $D(x)$ is non-negative as $\phi$ is
  feasible for $w$ and $w'_{\phi}(u,v)$ does not introduce a negative cycle.
  Therefore, $\phi$ remains feasible without updating.

  Hence, assume that $\negbroken \leq 0$.
  Observe that $\phi'(x) = \phi(x)$ if $D(x) \geq \broken$.
  Thus, we only need to consider nodes $x$ which have a shortest
  \nodepath{v}{x} path of length at most $\broken$.
  This allows us to heavily prune the shortest path tree~$T$ by only
  including nodes $x$ with $D(x) \leq \broken$.
  Because $\phi$ is feasible for $w$ and due
  to~\cref{obs:new_weight_in_cycle,lem:infeasible_implies_negative_cycle},
  $T$ always exists and does not contain~$u$.

  We now show that all edges have non-negative edge
  weights~$w_{\phi'}({\cdot})$.
  Since we set $\phi'(v) = \phi(v) - \negbroken - D(v) = \phi(v) -
  \negbroken$, we get $w'_{\phi'}(u,v) = 0$ and $(u,v)$ is no longer
  \emph{broken}.
  Now consider any edge~$(x,y) \in E$, where the edge~$(x,y)$ is

  \begin{itemize}
    \item \ldots a tree edge in the shortest path tree $T$ (if the
      weight change does not introduce a negative cycle). Then, the
      new potential weight~$w'_{\phi'}(x,y)$ is
      {\small
        \begin{align}
          w'_{\phi'}(x,y)
          & = w_{\phi}(x,y) + [\broken - D(y)] - [\broken - D(x)]
          \nonumber                         \\
          & =w_{\phi}(x,y) + D(x) - D(y) \nonumber
          \\
          & = w_{\phi}(x,y) + D(x) - \left[D(x) +
          w_{\phi}(x,y)\right] = 0 \label{eq:edges_in_tree}
      \end{align}}

    \item \ldots not a tree edge, but both nodes $x$ and $y$ are in $T$.
      Observe that $D(x) + w_{\phi}(x,y) \geq D(y)$ and
      \cref{eq:edges_in_tree}: the new potential weight is non-negative.

    \item \ldots not in~$T$, and only $x$ is in~$T$. Then, we know
      \begin{equation*}
        D(x) + w_{\phi}(x,y)\ \                 \geq \ \ \broken
        \quad\quad\Leftrightarrow\quad\quad
        w_{\phi}(x,y)\ \  \geq \ \ \broken - D(x).
      \end{equation*}

      Thus $w'_{\phi'}(x,y) \ge 0$ after only adjusting $\phi(x)$
      since the decrease is less than the current value.
      Further, if we find $y$, but not $x$ the potential weight can
      only increase and thus not break the edge.
  \end{itemize}

  \noindent
  Finally, if $x, y \not \in T$, the non-negative weight
  $w'_{\phi'}(x,y) = w_\phi(x,y)$ stays unchanged.
\end{proof}

\begin{lemma}\label{cor:update_runtime}
  Checking if $\weightUp{w}{(u,v)}{c} \in \setCons$
  (Line~\ref{line:sampler_consistent} of Alg.~\ref{algo:sampler})
  runs in time $\Oh(m + n\log n)$.
\end{lemma}
\begin{proof}
  If an edge weight increases or $\negbroken \geq 0$, we directly
  accept keeping $\phi$ unchanged.
  Otherwise, compute a shortest path tree~$T$ from node $v$ and
  maximum distance \broken, \ie stop if every unvisited node $x$ has
  a shortest $\nodepath v x$ path of weight of at least \broken.
  If we find the target~$u$ with $D(u) < \broken$, reject the update
  as it causes a negative cycle (\cf
  \cref{obs:new_weight_in_cycle,lem:infeasible_implies_negative_cycle}).
  If $u$ is not found, \ie $D(u) \ge \broken$, accept and update
  $\phi$ as in~\cref{lem:potential_update}.
  Since $\phi$ is feasible before, every edge has a non-negative
  potential weight, and we can use \algdk to compute $T$ or find a
  path that leads to rejection.
\end{proof}

\subsection{Bidirectional search}\label{subsec:bidirectional}
The algorithm outlined in the previous section executes \algdk to
obtain shortest path distances to nodes with distance at most $\broken$.
In practice, we can often further prune the search space:
For exposition, assume we sampled weight~$c$ for edge~$(u,v)$
introducing a negative cycle.
We detect this, by finding a $\nodepath v u$ path~$P$ with $w(P) + c
< 0 \Leftrightarrow w(P) < -c$.
We may have found~$P$ faster through a bidirectional search, \ie we
run a forward search from node $v$, and a backward search from $u$
(reversing each edge).
Then, interleaving the forward- and backward searches eventually
yields a node~$x$ with $w_{\phi}(\nodepath v x) + w_{\phi}(\nodepath
x u) < \broken$.
It is known both in theory and practice, that this can reduce the
search space
dramatically~\cite{DBLP:journals/jea/BorassiN19,DBLP:conf/iwoca/BlasiusW23,DBLP:journals/talg/BlasiusFFKMT22}.

As illustrated in \cref{fig:potentials} (Solution 2), the
bidirectional search also often accelerates accepted updates,
especially if there is a broken edge.
Consider an update of $w(u,v)$ that breaks the edge and requires
increasing the gradient $\phi(v) - \phi(u)$ by $\broken$.
Following the unidirectional pattern of \algdk, we previously only
increased~$\phi(u)$ and kept $\phi(v)$.
The main insight is that we can achieve the same effect by
distributing the changes over both endpoints, \ie by reducing
$\phi(u) - \Delta_u$ and increasing $\phi(v) + \Delta_v$ with
$\Delta_u + \Delta_v = \broken$.

This strategy directly repairs the broken edge~$(u,v)$ and may still
break out-edges of $v$ (and successors).
As before, these cascading defects can be dealt with using
\cref{lem:potential_update}.
Additionally, the new strategy can also break in-edges of $u$ (and
predecessors).
These can be fixed through a mirrored variant of
\cref{lem:potential_update} by decreasing (rather than increasing)
potentials according to $\nodepath x u$ paths (rather than $\nodepath
v x$ paths).
We give empirical evidence in \cref{sec:experiments} that a
bidirectional search results in a significantly faster algorithm.

\begin{lemma}\label{lem:potential_rev_update}
  For a weight decrease of edge $(u,v)$, let $w'_{\phi}(u,v)$ be the
  new potential weight.
  If $w'$ remains free of negative cycles, we can construct a
  feasible $\phi'$ with $\phi'(x) = \phi(x) - \max\{0,
  -w'_{\phi}(u,v) - D(x)\}$ where $D(x)$ is the length of a shortest
  \nodepath{x}{u} path with respect to $w_{\phi}$.
\end{lemma}

\begin{proof}
  This lemma is the symmetric version of \cref{lem:potential_update}.
  The proof is analogous to the proof of~\cref{lem:potential_update}
  with the inverted graph and potentials.
  \hfill
\end{proof}

The distribution between $\Delta_u$ and $\Delta_v$ does not affect correctness.
Our implementation simply alternates between forward-- and backward
search until it is established that there is no $\nodepath v u$ path
of length at most $\broken$.
This also yields the disjoint partial shortest path trees $T_u$
around $u$, and $T_v$ respectively, used for potential updates.

\subsection{Rejecting the rejection}\label{subsec:rej_rej}

As we will later see, the current chain eventually settles on a total
rejection rate exceeding $40\%$ in all parameter settings.
Hence, in almost $50\%$ of cases, we ``discard'' our search and do
not change $w$, which significantly slows convergence, as we do
nothing while incurring the cost of a search.
We, however, can prevent this behavior entirely by extending
\cref{algo:sampler} to include an additional re-sampling step of $c$
\emph{after} ``rejecting'' in line $5$.
We again point to \cref{obs:new_weight_in_cycle} and note that we
accept any new weight $c$ for $e = (u, v) \in E$ iff $c \in \wInt_e$
where $\wInt_e$ is defined below.

\begin{definition}\label{def:acceptable_wint}
  For edge $e = (u, v) \in E$ and consistent weights $w$, define the
  set of acceptable weights as $\wInt_{e} \coloneqq \setc{x \in
  \wInt}{x \geq -w(P)}$ where $P$ is the shortest \nodepath{v}{u}
  path in $(G, w)$.
\end{definition}

Hence, to resample an acceptable new weight $c'$ for $w(e)$, we can
use the path length $w(P)$, which we computed to reject the
originally proposed weight $c$, to lower bound $c'$.
We dub this altered scheme \norej.

Unfortunately, in practice, one might not compute the true shortest
\nodepath{v}{u} path $P$ in $(G, w)$ (and thus $\wInt_e$) to reject
$c$ but rather any \nodepath{v}{u} path $P'$ with $c < -w(P')$ (\eg
reject at the first observed smaller path).
Therefore, to guarantee that the newly resampled weight $c'$ also
does not induce a negative cycle, we continue with the (rejecting)
search in two ways:
\begin{itemize}
  \item[(i)] Continue the search until the truly shortest
    \nodepath{v}{u} path $P$ is found and resample $c'$ once
    afterwards in the correct set $\wInt_e$; or:
  \item[(ii)] Directly resample $c'$ in a superset of $\wInt_e$ and
    continue the search until you either accept or reject and
    resample $c'$ again with now better bounds before repeating this procedure.
\end{itemize}
While option (i) only requires one resampling, it also requires us to
always find the shortest \nodepath{v}{u} path, which can be expensive.
On the other hand, option (ii) might need multiple resampling steps,
but also still allows for early termination by acceptance before
fully discovering $P$.\footnote{
  Implementations of both methods are practically identical.
  Option (i) replaces the `resampling' step by always `resampling'
  the lower bound --- guaranteeing a continuance of the search until
  $P$ is found.
}
We will see in \cref{sec:experiments} that option (ii) is by far the
most efficient version for our process.
Note that in all cases (original chain, option (i), option (ii)), we
only accept iff the (final) sampled weight lies in $\wInt_e$.
We dub option (i) and (ii), \mincyc and \resamp, respectively.

For completeness, we briefly prove that these variations also
converge to the uniform distribution on $\setCons_{G, \wInt}$.
\begin{theorem}\label{thm:convergence-mc-no-rej}
  Let $G = (V, E)$ be a directed graph and $\wInt$ a weight interval.
  Then, the \norej~process converges to a uniform distribution on
  $\setCons_{G, \wInt}$.
\end{theorem}
\begin{proof}
  The proof follows the same steps as the proof of~\cref{thm:convergence-mc}.
  Since the process operates on the same state graph as
  \cref{algo:sampler} and only varies in transition probabilities,
  properties of \underline{irreducibility} and
  \underline{aperiodicity} still hold, and we only need to show
  \underline{symmetry} again:

  Let $w_1, w_2 \in \setCons, w_1 \neq w_2$ with $w_1 \in N[w_2]$.
  There must exist a unique $e = (u, v) \in E$ with $w_1(e) \neq w_2(e)$.
  Since $w_1$ and $w_2$ only differ in $e$, and both yield consistent
  weight functions, the weight of the shortest \nodepath{v}{u} path
  $P$ is identical in $(G, w_1)$ and $(G, w_2)$, and both weights
  share the same $\wInt_e$.
  Furthermore, by \cref{obs:new_weight_in_cycle}, we have $w_1(e),
  w_2(e) \in \wInt_e$.

  By definition of \norej, the new weight $c$ (or $c'$) is a uniform
  sample in a superset of $\wInt_e$, and thus also a uniform sample
  in $\wInt_e$ if it gets accepted.
  Thus, to transition from $w_1$ to $w_2$ ($w_2$ to $w_1$), we must
  first sample $e$ with probability $1/m$ and then $w_2(e)$
  ($w_1(e)$) with probability $1/|\wInt_e|$ (even while  possibly not
  knowing $|\wInt_e|$).
\end{proof}

\section{Empirical evaluation}\label{sec:experiments}
We empirically study our chains and the properties of the produced output.
We also give performance indications on the previously discussed
implementations on random graph models.
All performance-critical sections are implemented in
\textsc{Rust}\footnote{Compiled with \texttt{rustc-1.96.0-nightly}.
Performance roughly on par with C~\cite{DBLP:conf/clei/CostanzoRNG21}.}.
All runtime measurements were taken on a server with an \texttt{AMD
EPYC 7702P 64-Core} processor with $512$GB of RAM.
Pseudo-random numbers are generated with the \texttt{Pcg64}
generator~\cite{pcgrng}.

\subparagraph*{Parameters}
The implementations are relatively straightforward, with the
exception of avoiding negative cycles.
Here, we implement both the \algdk and \algbd variants of \rej as
well as both discussed variants of \norej.
Since applying~\cref{lem:potential_update} leads to a high number of
edges with \emph{potential weight} $0$, we modify \algdk as follows:
If we settle a node, we also directly settle all neighbors reachable
by zero-weight edges without pushing them into the priority queue.\footnote{
  In initial experiments, this heuristic proved to be \textbf{only}
  effective for \algdk.
}

We consider Gilbert's $\mathcal{G}(n, p)$
model~\cite{gilbert1959random} with $p = \avgdeg / n$ to simulate a
graph with average degree $\avgdeg$ which we dub $\gnp(n, \avgdeg)$
and the hyperbolic graph model $\rhg(n, \avgdeg)$~\cite{gen/rhg} with
dispersion parameter $\alpha = 0$ and $T = 0$, where we replace an
undirected edge by two directed edges.
We also derive a real-world dataset \roads from US state road
networks obtained from the TIGER/Line dataset, available on the 9th
DIMACS Implementation Challenge
website.\footnote{\url{https://www.diag.uniroma1.it/challenge9/data/tiger/}}
We replace networks with $n\geq \num{100000}$ with a connected
induced subgraph of $\num{100000}$ nodes, collected via BFS.
We also replace each undirected edge by two directed edges.

We implement four different initial consistent weight functions:
\wmax (start with all weights set to $\max(\wInt)$),
\wzero (start with all weights set to $0$),
\wone (start with all weights set to $1$),and
\wunif (weights are independently drawn uniformly from $[0,\max(\wInt)]$).
Unless stated otherwise, we set $\wInt = [-100,100] \cap \mathbb{N}$
and $n = \num{10000}$ for \gnp and \rhg graphs.
All averages are computed via the arithmetic mean.

\subsection{Selecting the best implementation}
Our main goal is to study the performances of our different
algorithms and chains to find the implementation that yields an
approximate uniform distribution the fastest.
For that, we first compare the different search algorithms for \rej
before comparing them against implementations of \norej.
For that, we ran experiments on \gnp and \rhg graphs with average
degrees $\avgdeg \in \set{10, 20, 50}$ as well as graphs from $\roads$.
We start with all four different initial weight functions, repeating
each datapoint $100$ times.

\begin{table}[h!]
  \caption{Average speedup of \algbd over \algdk after $10^8$ rounds (\rej).}
  \label{tab:cmp-search}

  \begin{center}
    \begin{tabular}{c||c|c|c|c|c|c|c}
      \textbf{Graph} & \multicolumn{3}{c|}{$\gnp(n = 10^5)$} &
      \multicolumn{3}{c|}{$\rhg(n = 10^5)$} & \multirow{2}{*}{\roads}
      \\\cline{1-7}
      \textbf{Degree} & $10$ & $20$ & $50$ & $10$ & $20$ & $50$ \\ \hline \hline
      Speedup & $27.3052$ & $29.0018$ & $25.8080$ & $13.9757$ &
      $28.5909$ & $38.7055$ & $1.2643$
    \end{tabular}
  \end{center}
\end{table}

\subparagraph*{Optimal Algorithm}
First, we want to verify whether a bidirectional search via \algbd
does yield a better runtime than a one-directional search via \algdk.
The implicit assumption that a naive implementation of \rej using
\algbf performs by far the worst was confirmed in initial experiments
and is thus no longer considered.

\Cref{tab:cmp-search} supports our hypothesis as in almost all
instances, \algbd averages a speedup\footnote{Speedup is computed via
the total running time (all previous rounds) thus far.}
of more than $25$ over \algdk, showing a significant runtime improvement.
For graphs in $\roads$, we observe a smaller speedup of roughly
$25\%$ since these graphs are very sparse, leading to narrow search
trees in both algorithms.

\begin{table}[h!]
  \caption{
    Comparison of \algbd (\rej), \resamp, and \mincyc after $10^8$ rounds:
    \textsc{AccRate} is the acceptance rate of \rej, \textsc{Speedup}
    the speedup of \resamp and \mincyc over \algbd, and
    \textsc{AccSpeedup} the speedup per accepting round (\ie
    $\textsc{Speedup} / \textsc{AccRate}$). \textsc{Resamples} is the
    average number of times \resamp and \mincyc had to `resample' in a search.
  }
  \label{tab:cmp-chains}

  \begin{center}
    \begin{tabular}{c|c||c|c|c|c|c|c|c}
      \multicolumn{2}{c||}{\textbf{Graph}} &
      \multicolumn{3}{c|}{$\gnp(n = 10^5)$} &
      \multicolumn{3}{c|}{$\rhg(n = 10^5)$} & \multirow{2}{*}{\roads}
      \\\cline{1-8}
      \multicolumn{2}{c||}{\textbf{Degree}} & $10$ & $20$ & $50$ &
      $10$ & $20$ & $50$ \\ \hline \hline
      \multicolumn{2}{c||}{\textsc{AccRate} (\rej)} & $0.5921$ &
      $0.5472$ & $0.5176$ & $0.5465$ & $0.5246$ & $0.5102$ & $0.6557$ \\ \hline
      \multirow{2}{*}{\textsc{Speedup}} & \resamp & $0.8620$ &
      $0.8980$ & $0.8655$ & $0.8392$ & $0.8374$ & $0.8161$ & $0.8580$ \\
      & \mincyc & $0.2861$ & $0.2732$ & $0.2499$ & $0.3199$ &
      $0.2401$ & $0.1898$ & $0.6959$ \\ \hline
      \multirow{2}{*}{\textsc{AccSpeedup}} & \resamp & $1.4561$ &
      $1.6413$ & $1.6722$ & $1.5363$ & $1.5974$ & $1.5996$ & $1.3087$ \\
      & \mincyc & $0.4832$ & $0.4992$ & $0.4829$ & $0.5855$ &
      $0.4578$ & $0.3720$ & $1.0612$ \\ \hline
      \multicolumn{2}{c||}{\textsc{Resamples} (\resamp)} & $1.2485$ &
      $1.2539$ & $1.2550$ & $1.1595$ & $1.1890$ & $1.2090$ & $1.0150$
    \end{tabular}
  \end{center}
\end{table}

\subparagraph*{Optimal Chain}
Similarly, we also want to determine whether eliminating rejections
from \rej yields any practical benefit.
While we do expect the average runtime per round to go up, as we
potentially do more work per round (resampling in case of rejection,
continuing the search), we anticipate that the average runtime per
accepted weight change goes down.
We show later that this is a reasonable metric as the asymptotic
behaviour of the underlying markov chain is identical under accepted
weight changes.
Given the previously observed superiority of \algbd, we consider only
a bidirectional approach going forward.

\Cref{tab:cmp-chains} confirms our hypotheses.
Firstly, more than one third of all rounds in \rej were rejected,
sometimes even up to almost $50\%$.
This shows that the underlying Markov chain of \rej is almost lazy in
most cases.
We also confirm that the additional overhead of resampling and
continuing the search does slow down \resamp and \mincyc initially by
a factor of $\sim0.85$ and $\sim0.25$, respectively.
However, when averaging the runtime over accepted rounds, this flips
for \resamp, which now achieves a roughly $50\%$ speedup over \rej.
Nonetheless, even then, \mincyc is still circa twice as slow as \rej.
The number of average resamples needed reaffirms this as \resamp on
average succeeds with almost only one resample.

Combining results from \cref{tab:cmp-search,tab:cmp-chains}, for a
metric of accepted rounds, a bidirectional \resamp is more than
$30\times$ faster than the basic \algdk implementation for \rej in
most cases --- barring an even slower naive implementation of \algbf.
Lastly, for $n = 10^5$, \resamp only requires $4.6\mu s \sim 5.6\mu
s$ per round for \gnp graphs and $0.7\mu s \sim 1.2\mu s$ per round
for \rhg graphs on average.
For $\roads$, each round only takes $1\mu s$ on average.

\subsection{Output Metrics}
Apart from the runtime per round, the number of total rounds to reach
an approximately uniform distribution is of interest.
In practice, the \emph{mixing time} of the MCMC is often investigated
by considering a collection of metrics which are regarded as proxies
for the mixing state of the underlying chain.

We consider two metrics as proxies: (i) the average weight of the
graph and (ii) the fraction of edges that have negative weights.
We again consider \gnp and \rhg graphs with $n = 10000$ and $\avgdeg
\in \set{10, 20, 50}$ and run the chain for $2.1 \cdot 10^8$ ($10^8$
for \roads) rounds.

\begin{figure}[tb]
  \centering
  \begin{subfigure}{0.5\textwidth}
    \centering
    \includegraphics[width=0.9\textwidth]{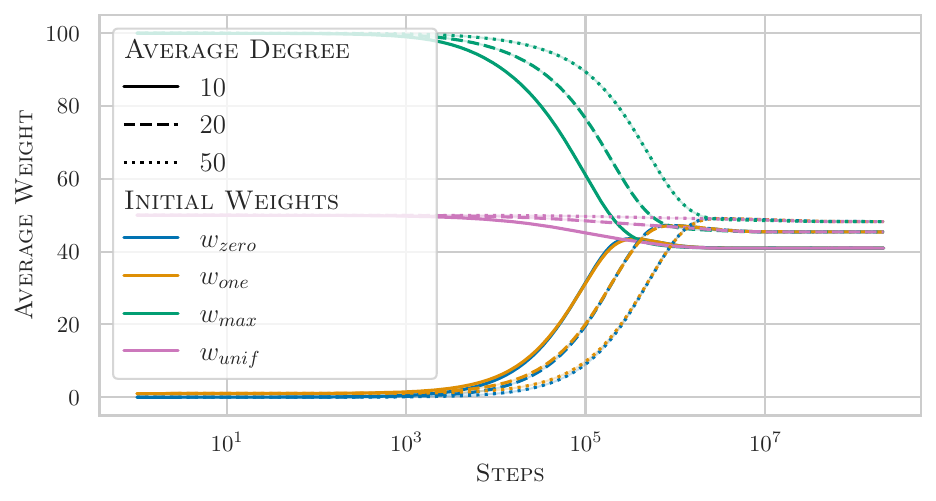}
  \end{subfigure}%
  \begin{subfigure}{0.5\textwidth}
    \centering
    \includegraphics[width=0.9\textwidth]{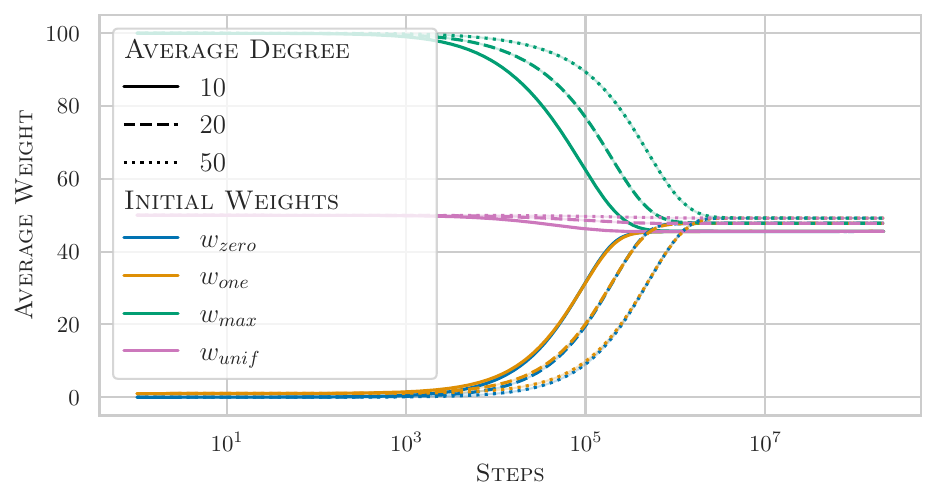}
  \end{subfigure}
  \caption{
    Average weight over time for $\gnp$ (l), $\rhg$ (r).
  }
  \label{fig:avg_weight}
\end{figure}

\begin{figure}[tb]
  \begin{subfigure}{0.5\textwidth}
    \centering
    \includegraphics[width=0.9\textwidth]{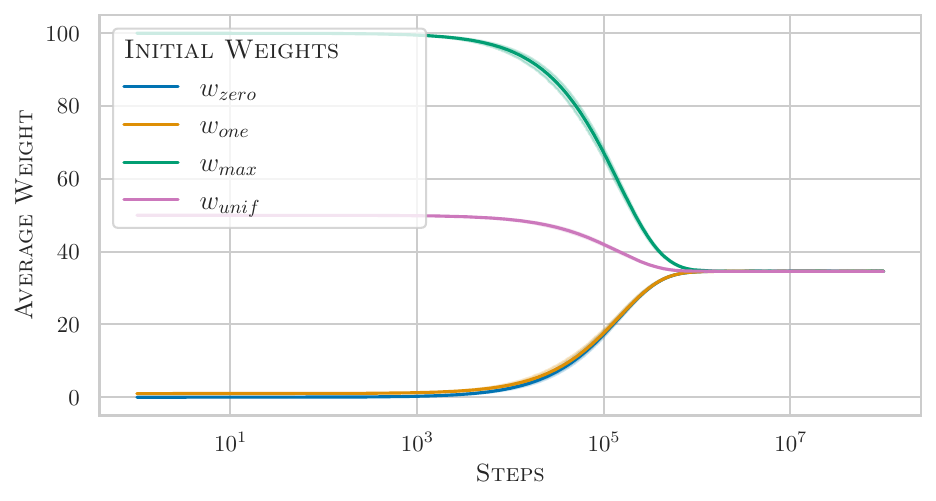}
  \end{subfigure}
  \begin{subfigure}{0.5\textwidth}
    \centering
    \includegraphics[width=0.9\textwidth]{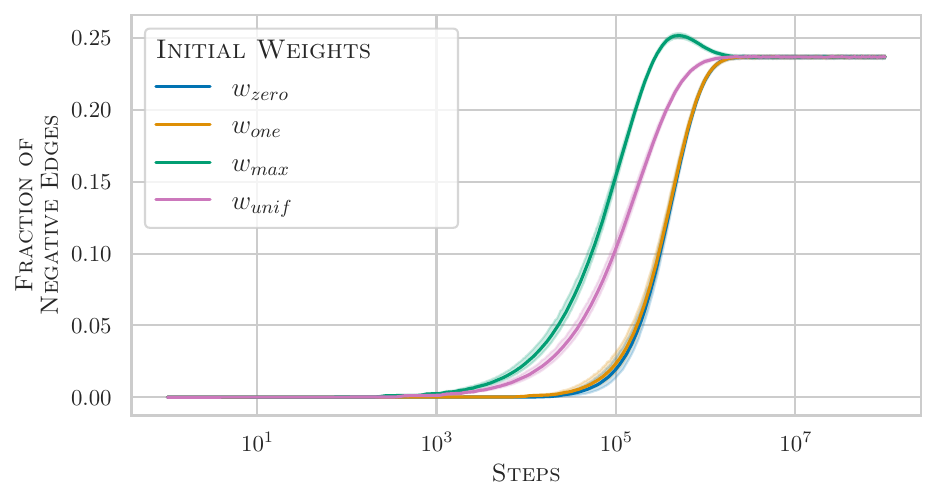}
  \end{subfigure}
  \caption{
    Average weight (l) and fraction of negative edges (r) over time
    on $\roads$ networks.
  }
  \label{fig:road_networks}
\end{figure}

\subparagraph*{Average Weight}
We first consider the average total weight of all edges over time.
\Cref{fig:avg_weight} shows that the average weight stabilizes after
at most $2\cdot10^6$ steps for both graph classes.
There is also a slight bias towards a higher average weight if the
average degree is higher.
This is to be expected as a higher average degree equates to a higher
number of cycles and thus more restrictions on negative weights.
\begin{figure}[tb]
  \centering
  \begin{subfigure}{0.5\textwidth}
    \centering
    \includegraphics[width=0.9\textwidth]{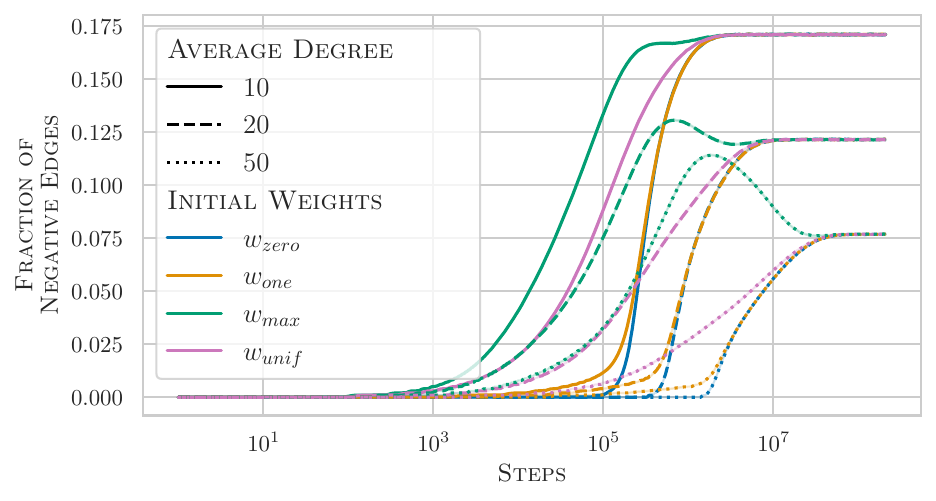}
  \end{subfigure}%
  \begin{subfigure}{0.5\textwidth}
    \centering
    \includegraphics[width=0.9\textwidth]{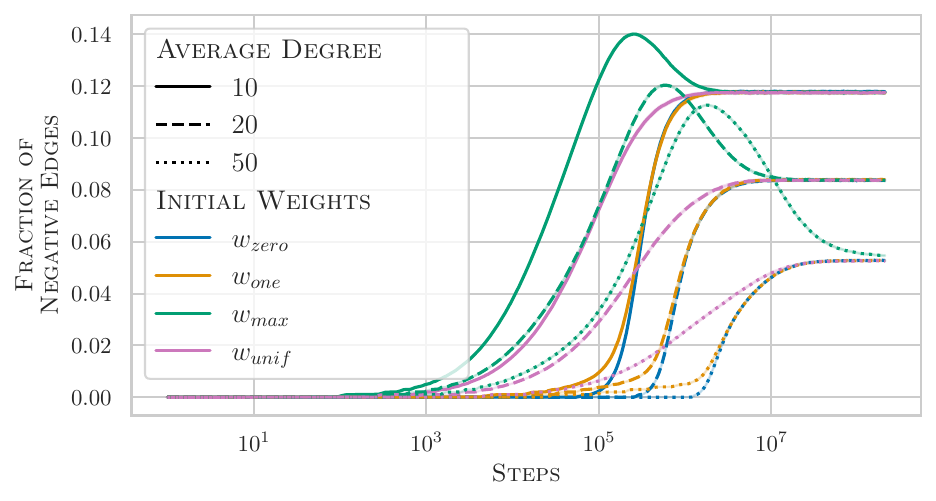}
  \end{subfigure}
  \caption{Fraction of negative edges over time for $\gnp$ (l), $\rhg$ (r).}
  \label{fig:negfrac}
\end{figure}

\subparagraph*{Fraction of Negative Edges}
We also consider the fraction of negative edges over time.
\Cref{fig:negfrac} summarizes the results.
Similarly to the average weight, we observe both a bias towards less
negative edges for bigger degrees as well as a stabilizing effect ---
although now after $10^6\sim10^8$ rounds, depending on average degree
and graph type.
Interestingly, for starting weights $\wmax$, the fraction of negative
edges initially ``overshoots'' before settling down again into a stable state.
This effect can be explained in part by the residual high-weight
edges that have not yet been chosen as the target of an update. The
chains require~$\Theta(m \log m)$ rounds to ensure that each edge is
selected at least once with high probability. Since for all
considered $\gnp$ and $\rhg$ instances we have $10^6 \leq m \log m
\leq 10^7$ it follows that, while these high-weight edges persist,
they permit a significantly larger proportion of negative edges
within any cycle containing them. Nonetheless, there still remains a
lingering, slowly vanishing effect, well beyond the stated $m\log m$ bound.
This effect seems to be driven by high-weight edges that persist well
into the process and are prevented to significantly decrease their
weight due to emerging constraints imposed by the, on average,
decreasing cycle weight. This argument derives additional credence
from two facts: (i) By \cref{fig:avg_weight,fig:road_networks} the
average weight stabilizes well before the fraction of negative edges
and (ii) the expected edge weight in absence of constraints is $0$ by
the choice of $\wInt$, hence as soon as the average weight converges
(bounded away from $0$), constraints seem prevent a further decrease.
The former suggests that there is still a significant amount of
high-weight edges when average weight convergence is reached while
the later indicates that these high-weight edges are unlikely to
decrease their weight quickly due to the significant amount of
constraints imposed by the (low-weight) cycles at this point. Despite
this, in all observed cases, the chain stabilizes in $o(m^2)$ rounds.

\subsection{Fast convergence on the
\boldmath$n$-cycle}\label{subsec:experiment_rapid_mixing}
Here, we aim to provide empirical evidence that, in practice,
comparatively low mixing times suffice for the $n$-cycle and $\wInt =
\set{-1, 0, 1}$.

\begin{figure}
  \centering
  \begin{subfigure}{0.5\textwidth}
    \centering
    \includegraphics[width=0.9\textwidth]{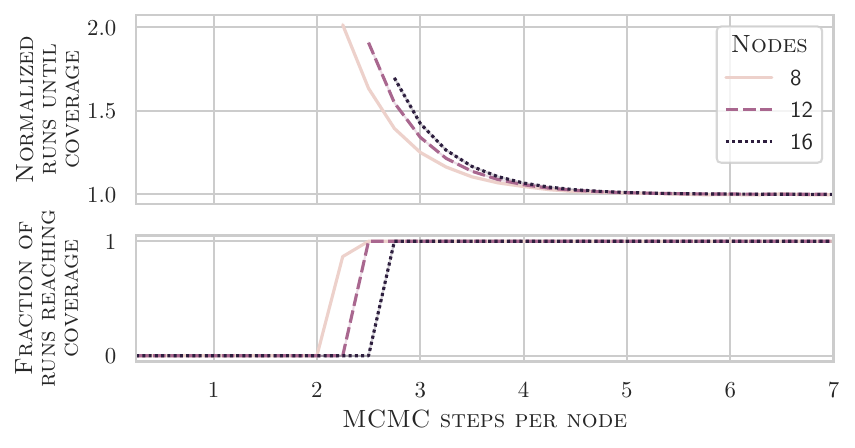}
  \end{subfigure}%
  \begin{subfigure}{0.5\textwidth}
    \centering
    \includegraphics[width=0.9\textwidth]{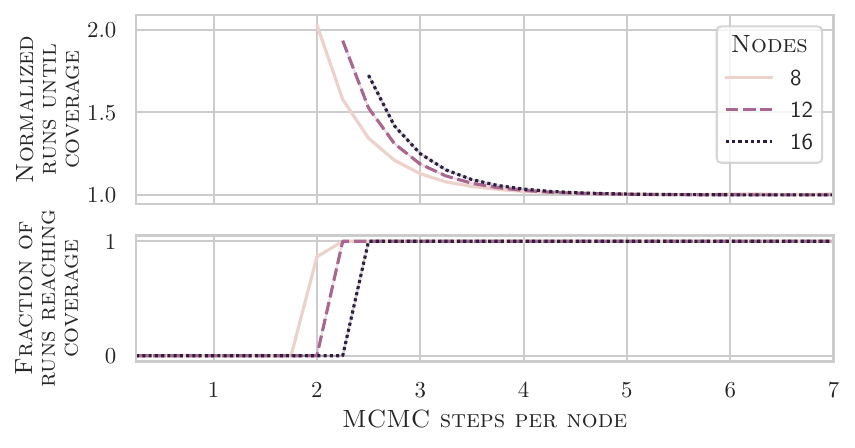}
  \end{subfigure}
  \caption{
    Runs to see $99\%$ of $\setCons_{C_n, \wInt}$ on $n$-cycle for
    $\wInt = \set{-1, 0, 1}$ as function of MCMC steps for $\rej$ (l)
    and $\norej$ (r).
  }
  \label{fig:cycle}
\end{figure}

\subparagraph*{Methodology}
The experimental design is inspired by the Coupon Collector's problem.
We obtain a large number~$k$ of independent samples by running
\mcsteps~steps of the Markov Chain each starting from the initial
weight $(0, \ldots, 0)$.
Since each run is independent, some obtained weights are inevitably duplicates.
However, we know from the Coupon collector's problem that for
large~$n$, $k = \Theta(|\setCons_{C_n,\wInt}|
\log(|\setCons_{C_n,\wInt}|))$ samples are expected to be sufficient
to observe every item in $\setCons_{C_n,\wInt}$ at least once.
Since the last few elements incur significant noise, we stop earlier,
namely when $99\%$ of elements were seen.\footnote{We also ensure
  that all samplers can \emph{fully} cover~$\setCons_{C_n, \wInt}$ to
give additional  credence to the correctness.}
We call this quantity \emph{runs-until-coverage}.

\subparagraph*{Measurements}
We use an exact sampling algorithm\footnote{
  We draw each edge weight~$w(e)$ i.i.d. from $\wInt = \set{-1, 0,
  1}$, and restart if $\sum_e w(e)$ is negative.
  This yields a constant acceptance probability on the $n$-cycle.
} to establish the correct baseline of
$(4.6 \pm 0.1) |\setCons_{C_n,\wInt}|$ runs-until-coverage for $n \in
\set{8,12,16}$.
This linear dependency is plausible since each sample is a new item
with probability of at least $1/100$.
Then, we replace the exact sampler with the MCMC process for
different numbers of steps~\mcsteps and measure their runs-until-coverage.
This allows us to reject the hypothesis that the MCMC process is
identical to the exact sampler if \mcsteps is too small.
If coverage does not occur within $10 |\setCons_{C_n,\wInt}|$ steps
(\ie more than twice the baseline), we say that the run does not reach coverage.
\Cref{fig:cycle} contains the findings as function of~\mcsteps and
normalized by the exact sampler's runs-until-coverage.

Between $2n \le \tau \le 3n$ steps, we observe a sharp transition
from no runs achieving coverage, to all runs completing successfully
for both chains.
Observe that $\norej$ reaches coverage slightly earlier than $\rej$
since the MCMC steps of $\rej$ contains both successes and rejections
while $\norej$ by design has only successful rounds.
This provides strong evidence, at least on the $n$-cycle, that $\rej$
and $\norej$ display a similar convergence rate.
Comparing $n=8$ and $n=16$, we see a slight increase in
runs-until-coverage that is consistent with a mixing time of $\Theta(n \log n)$.

\subsection{Identical Convergence}
The previous experiment is only feasible for very small graphs and
thus not practical for our main datasets.
Thus, we want to provide further evidence that both \rej and \norej
(i) converge relatively fast and (ii) converge at the same rate when
only counting accepting rounds in \rej.
Hence, we run statistical tests on meta-distributions of weights over
time and show that convergence behaviour is practically identical for
\rej and \norej.

\subparagraph*{Weights}
We restrict ourselves to \gnp and \rhg graphs with $n = 100$ and
average degrees $\avgdeg \in \set{10, 20}$.
We then simulate both \rej (\algbd) and \norej (\resamp) on each
graph $n \cdot \avgdeg \cdot 10$ $(\in \set{\num{10000},
\num{20000}})$ times per possible starting weight ($\wmax, \wzero,
\wone, \wunif$).
After $2^i$ \textbf{accepting} rounds in each chain, we take a
snapshot of the entire current weight vector of the graph.

\subparagraph*{Analysis}
For each type of graph and time step $s$, we have up to
$\num{40000}/\num{80000}$ weight distributions of length
$\num{1000}/\num{2000}$ and thus a distribution over distributions.
To compare two such distributions using common statistical tests, we
need to find a representative value for each weight vector yielding a
distribution over representatives.
Natural examples are minimum, maximum, mean, and taking the weight of
a specific edge (index).
\begin{figure}[!htb]
  \centering
  \begin{subfigure}{0.5\textwidth}
    \centering
    \includegraphics[width=0.9\textwidth]{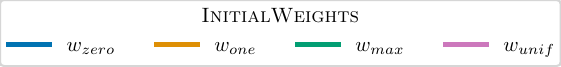}
  \end{subfigure}
  \begin{subfigure}{0.9\textwidth}
    \centering
    \includegraphics[width=0.9\textwidth]{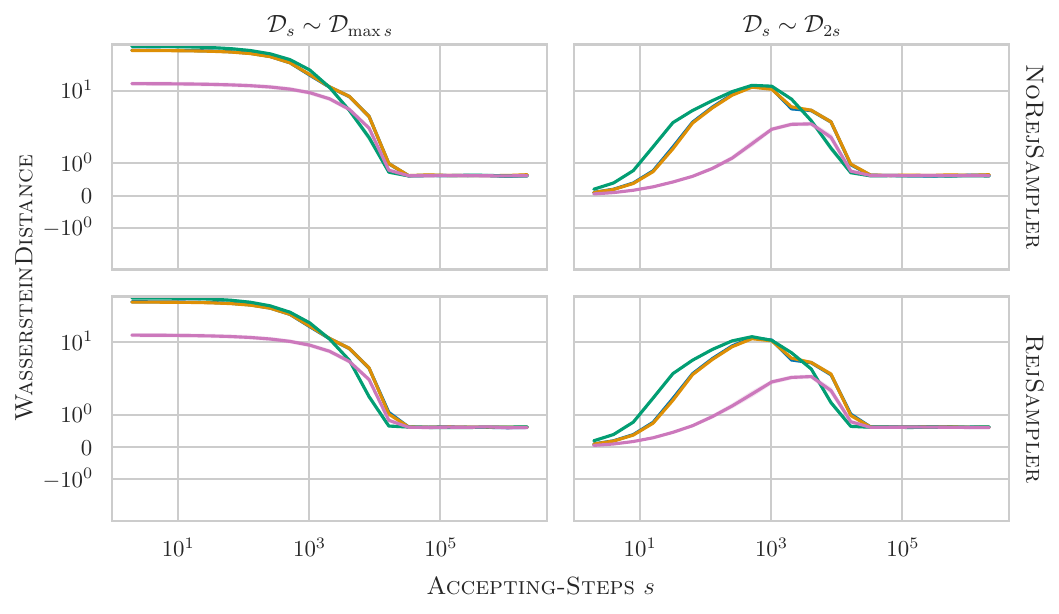}
  \end{subfigure}
  \begin{subfigure}{0.9\textwidth}
    \centering
    \includegraphics[width=0.9\textwidth]{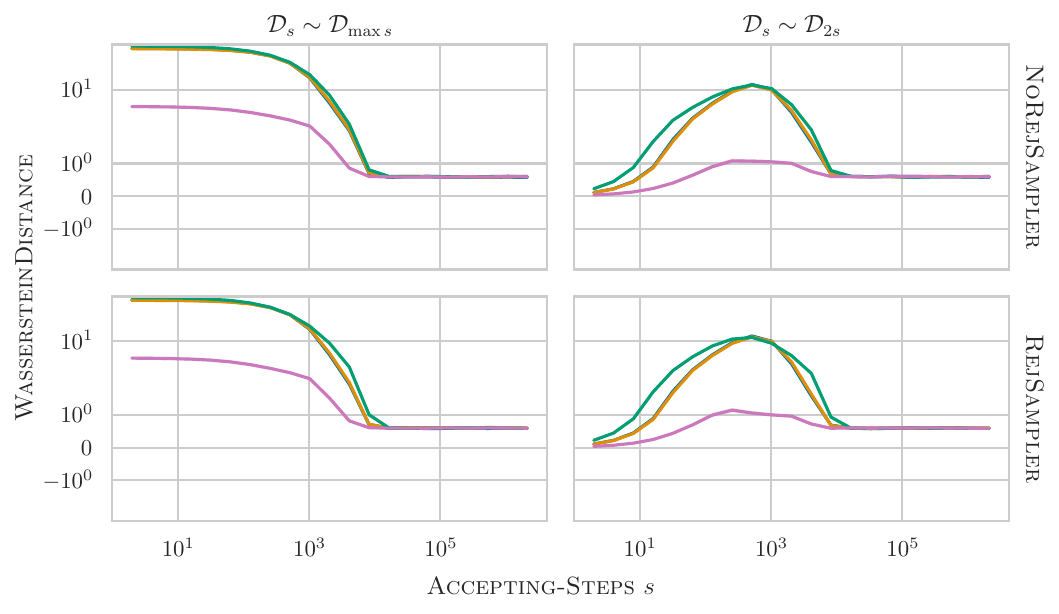}
  \end{subfigure}
  \caption{%
    Wasserstein-Distances for weight-distributions of $\gnp(n = 100,
    \avgdeg = 10)$ (top) and $\rhg(n = 100, \avgdeg = 10)$ (bottom) graphs.
    Weights are categorized by sampler and initial weight function;
    meta-distributions are then constructed for each edge (index),
    evaluated, and then aggregated in these plots.
    \vspace{-1em}
  }
  \label{fig:ws_stats}
\end{figure}

We then compare two meta-distributions using statistical tests to
determine how close they are.
We do this in two ways: (i) comparing distributions at step $s$ with
distributions at step $2s$, and (ii) comparing distributions at step
$s$ with distributions at step $s_{max}$ where $s_{max}$ is the
highest recorded step.
(i) determines how many changes occur between $s$ accepting rounds,
whereas (ii) determines how close a distribution at step $s$ is to a
final distribution at step $s_{max}$, which we assume to be a fairly
good approximation of a uniform distribution.

\subparagraph*{Results}
\Cref{fig:ws_stats} shows the results of a statistical analysis using
the Wasserstein distance~\cite{wasserstein} for graphs with average
degree $\avgdeg = 10$ (see Appendix for $\avgdeg = 20$: \Cref{fig:ws_stats_20}).
Although different types of graphs, we observe an almost
indistinguishable behaviour in all cases.
For the comparison of $s$ and $s_{max}$ (left), meta-distributions
differ entirely from the final target distribution before converging
around $10^4$ accepting steps.
On the other hand, when comparing distributions at $s$ and $2s$ steps
(right), we see a clear trend where initially similar distributions
first stray away from each other before converging again at $10^4$ steps.

In both cases, distributions start converging after $10^3$ steps,
coinciding with the number of edges.
After circa $10^4$ steps, the process enters a stable state.
While this is not a full measure of the converged mixing time of the
chains, it does show that the two chains, in fact, converge at a
similar rate when measured on a metric of accepting steps.
Paired with previous results, this reaffirms that (i) the chain
converges relatively fast in practice and that (ii) \resamp is the
most efficient algorithm to sample uniform negative edge weights.

\section{Summary and Future Work}
We introduce a maximum entropy model for signed edge weights,
propose an MCMC sampler, and show that it converges to the model's distribution.
We then engineer an implementation for the MCMC using a bidirectional
search and a novel resampling technique.
We empirically study the model, the sampler's convergence, as well as
trade-offs in the implementation to provide an optimized,
ready-to-use generator for uniform negative edge weights.

Future work might entail better theoretical bounds on the mixing time
of our chains.
The upper bound for the $n$-cycle, while polynomial, is highly
impractical and not in line with experimentally verified bounds in
\cref{subsec:experiment_rapid_mixing}.

It might also be interesting to further investigate the impact of
asymmetric weight intervals (around $0$); while initial experiments
indicated no significant impact on running time, convergence, or even
some output metrics, a more rigorous evaluation is needed.

Finally, it might also be interesting to study how recent advances in
negative weight shortest path algorithms fare on weights generated by
our maximum entropy model.

\bibliography{negative_edge_weights,dblp}

@book{DBLP:books/cu/H2016,
  author       = {Remco van der Hofstad},
  title        = {Random Graphs and Complex Networks},
  series       = {Cambridge Series in Statistical and Probabilistic Mathematics},
  volume       = {43},
  publisher    = {Cambridge University Press},
  year         = {2016}
}

@inproceedings{DBLP:conf/wea/CassisKNR25,
  author       = {Alejandro Cassis and
                  Andreas Karrenbauer and
                  Andr{\'{e}} Nusser and
                  Paolo Luigi Rinaldi},
  title        = {Algorithm Engineering of {SSSP} with Negative Edge Weights},
  booktitle    = {{SEA}},
  series       = {LIPIcs},
  pages        = {10:1--10:20},
  publisher    = {Schloss Dagstuhl - Leibniz-Zentrum f{\"{u}}r Informatik},
  year         = {2025}
}

@book{DBLP:books/daglib/0012859,
  author       = {Michael Mitzenmacher and
                  Eli Upfal},
  title        = {Probability and Computing: Randomized Algorithms and Probabilistic
                  Analysis},
  publisher    = {Cambridge University Press},
  year         = {2005}
}

@inproceedings{DBLP:conf/alenex/GkantsidisMMZ03,
  author       = {Christos Gkantsidis and
                  Milena Mihail and
                  Ellen W. Zegura},
  title        = {The Markov Chain Simulation Method for Generating Connected Power
                  Law Random Graphs},
  booktitle    = {{ALENEX}},
  pages        = {16--25},
  publisher    = {{SIAM}},
  year         = {2003}
}

@inproceedings{DBLP:conf/alenex/StantonP11,
  author       = {Isabelle Stanton and
                  Ali Pinar},
  title        = {Sampling Graphs with a Prescribed Joint Degree Distribution Using
                  Markov Chains},
  booktitle    = {{ALENEX}},
  pages        = {151--163},
  publisher    = {{SIAM}},
  year         = {2011}
}

@inproceedings{DBLP:conf/clei/CostanzoRNG21,
  author       = {Manuel Costanzo and
                  Enzo Rucci and
                  Marcelo R. Naiouf and
                  Armando De Giusti},
  title        = {Performance vs Programming Effort between Rust and {C} on Multicore
                  Architectures: Case Study in N-Body},
  booktitle    = {{CLEI}},
  pages        = {1--10},
  publisher    = {{IEEE}},
  year         = {2021}
}

@inproceedings{DBLP:conf/coco/Karp72,
  author       = {Richard M. Karp},
  title        = {Reducibility Among Combinatorial Problems},
  booktitle    = {Complexity of Computer Computations},
  series       = {The {IBM} Research Symposia Series},
  pages        = {85--103},
  publisher    = {Plenum Press, New York},
  year         = {1972}
}

@inproceedings{DBLP:conf/dexa/ArafatBDB20,
  author       = {Naheed Anjum Arafat and
                  Debabrota Basu and
                  Laurent Decreusefond and
                  St{\'{e}}phane Bressan},
  title        = {Construction and Random Generation of Hypergraphs with Prescribed
                  Degree and Dimension Sequences},
  booktitle    = {{DEXA} {(2)}},
  series       = {{LNCS}},
  volume       = {12392},
  pages        = {130--145},
  publisher    = {Springer},
  year         = {2020}
}

@inproceedings{DBLP:conf/iwoca/BlasiusW23,
  author       = {Thomas Bl{\"{a}}sius and
                  Marcus Wilhelm},
  title        = {Deterministic Performance Guarantees for Bidirectional {BFS} on Real-World
                  Networks},
  booktitle    = {{IWOCA}},
  series       = {{LNCS}},
  volume       = {13889},
  pages        = {99--110},
  publisher    = {Springer},
  year         = {2023}
}

@inproceedings{DBLP:conf/focs/BernsteinNW22,
  author       = {Aaron Bernstein and
                  Danupon Nanongkai and
                  Christian Wulff{-}Nilsen},
  title        = {Negative-Weight Single-Source Shortest Paths in Near-linear Time},
  booktitle    = {{FOCS}},
  pages        = {600--611},
  publisher    = {{IEEE}},
  year         = {2022}
}

@inproceedings{DBLP:conf/stacs/AmanatidisK23,
  author       = {Georgios Amanatidis and
                  Pieter Kleer},
  title        = {Approximate Sampling and Counting of Graphs with Near-Regular Degree
                  Intervals},
  booktitle    = {{STACS}},
  series       = {LIPIcs},
  volume       = {254},
  pages        = {7:1--7:23},
  publisher    = {Schloss Dagstuhl},
  year         = {2023}
}

@inproceedings{DBLP:conf/wae/DemetrescuFMN00,
  author       = {Camil Demetrescu and
                  Daniele Frigioni and
                  Alberto Marchetti{-}Spaccamela and
                  Umberto Nanni},
  title        = {Maintaining Shortest Paths in Digraphs with Arbitrary Arc Weights:
                  An Experimental Study},
  booktitle    = {{WAE}},
  series       = {{LNCS}},
  volume       = {1982},
  pages        = {218--229},
  publisher    = {Springer},
  year         = {2000}
}

@article{DBLP:journals/compnet/VigerL16,
  author       = {Fabien Viger and
                  Matthieu Latapy},
  title        = {Efficient and simple generation of random simple connected graphs
                  with prescribed degree sequence},
  journal      = {J. Complex Networks},
  volume       = {4},
  number       = {1},
  pages        = {15--37},
  year         = {2016}
}

@article{DBLP:journals/cpc/MartinR06,
  author       = {Russell A. Martin and
                  Dana Randall},
  title        = {Disjoint Decomposition of Markov Chains and Sampling Circuits in Cayley
                  Graphs},
  journal      = {Comb. Probab. Comput.},
  volume       = {15},
  number       = {3},
  pages        = {411--448},
  year         = {2006}
}

@article{DBLP:journals/csur/DrobyshevskiyT20,
  author       = {Mikhail Drobyshevskiy and
                  Denis Turdakov},
  title        = {Random Graph Modeling: {A} Survey of the Concepts},
  journal      = {{ACM} Comput. Surv.},
  volume       = {52},
  number       = {6},
  pages        = {131:1--131:36},
  year         = {2020}
}

@article{DBLP:journals/dam/FriezeG85,
  author       = {Alan M. Frieze and
                  Geoffrey R. Grimmett},
  title        = {The shortest-path problem for graphs with random arc-lengths},
  journal      = {Discret. Appl. Math.},
  volume       = {10},
  number       = {1},
  pages        = {57--77},
  year         = {1985}
}

@article{DBLP:journals/jacm/Johnson77,
  author       = {Donald B. Johnson},
  title        = {Efficient Algorithms for Shortest Paths in Sparse Networks},
  journal      = {J. {ACM}},
  volume       = {24},
  number       = {1},
  pages        = {1--13},
  year         = {1977}
}

@article{DBLP:journals/jea/BorassiN19,
  author       = {Michele Borassi and
                  Emanuele Natale},
  title        = {{KADABRA} is an ADaptive Algorithm for Betweenness via Random Approximation},
  journal      = {{ACM} J. Exp. Algorithmics},
  volume       = {24},
  number       = {1},
  pages        = {1.2:1--1.2:35},
  year         = {2019}
}

@article{DBLP:journals/make/MoosHASCP22,
  author       = {Janosch Moos and
                  Kay Hansel and
                  Hany Abdulsamad and
                  Svenja Stark and
                  Debora Clever and
                  Jan Peters},
  title        = {Robust Reinforcement Learning: {A} Review of Foundations and Recent
                  Advances},
  journal      = {Mach. Learn. Knowl. Extr.},
  volume       = {4},
  number       = {1},
  pages        = {276--315},
  year         = {2022}
}

@article{DBLP:journals/mp/CherkasskyG99,
  author       = {Boris V. Cherkassky and
                  Andrew V. Goldberg},
  title        = {Negative-cycle detection algorithms},
  journal      = {Math. Program.},
  volume       = {85},
  number       = {2},
  pages        = {277--311},
  year         = {1999}
}

@article{DBLP:journals/nm/Dijkstra59,
  author       = {Edsger W. Dijkstra},
  title        = {A note on two problems in connexion with graphs},
  journal      = {Numerische Mathematik},
  volume       = {1},
  pages        = {269--271},
  year         = {1959}
}

@article{DBLP:journals/siamcomp/Tarjan72,
  author       = {Robert Endre Tarjan},
  title        = {Depth-First Search and Linear Graph Algorithms},
  journal      = {{SIAM} J. Comput.},
  volume       = {1},
  number       = {2},
  pages        = {146--160},
  year         = {1972}
}

@article{DBLP:journals/socnet/EverettB14,
  author       = {Martin G. Everett and
                  Stephen P. Borgatti},
  title        = {Networks containing negative ties},
  journal      = {Soc. Networks},
  volume       = {38},
  pages        = {111--120},
  year         = {2014}
}

@article{DBLP:journals/socnet/SmithHKLBB14,
  author       = {Jason M. Smith and
                  Daniel S. Halgin and
                  Virginie Kidwell{-}Lopez and
                  Giuseppe (Joe) Labianca and
                  Daniel J. Brass and
                  Stephen P. Borgatti},
  title        = {Power in politically charged networks},
  journal      = {Soc. Networks},
  volume       = {36},
  pages        = {162--176},
  year         = {2014}
}

@article{DBLP:journals/talg/BlasiusFFKMT22,
  author       = {Thomas Bl{\"{a}}sius and
                  Cedric Freiberger and
                  Tobias Friedrich and
                  Maximilian Katzmann and
                  Felix Montenegro{-}Retana and
                  Marianne Thieffry},
  title        = {Efficient Shortest Paths in Scale-Free Networks with Underlying Hyperbolic
                  Geometry},
  journal      = {{ACM} Trans. Algorithms},
  volume       = {18},
  number       = {2},
  pages        = {19:1--19:32},
  year         = {2022}
}

@book{AdvComb,
  title     = {Advanced Combinatorics},
  author    = {Louis Comtet},
  year      = 1974,
  publisher = {Springer Dordrecht}
}

@book{Barabasi2016-np,
  title     = {Network Science},
  author    = {Barabasi, Albert-Laszlo},
  year      = 2016,
  month     = jul,
  publisher = {Cambridge University Press},
  address   = {Cambridge, England},
  language  = {en}
}

@article{Bel58,
  title     = {ON A ROUTING PROBLEM},
  author    = {Richard Bellman},
  year      = 1958,
  journal   = {Quarterly of Applied Mathematics},
  publisher = {Brown University},
  volume    = 16,
  number    = 1,
  pages     = {87--90},
  issn      = {0033569X, 15524485},
  url       = {http://www.jstor.org/stable/43634538},
  urldate   = {2023-10-16}
}

@phdthesis{CarstensPhd,
  title  = {{Topology of Complex Networks: Models and Analysis}},
  author = {Carstens, C. J.},
  year   = 2016,
  school = {RMIT University}
}

@article{DEMARTINO2018e00596,
  title    = {An introduction to the maximum entropy approach and its application to inference problems in biology},
  author   = {Andrea {De Martino} and Daniele {De Martino}},
  year     = 2018,
  journal  = {Heliyon},
  volume   = 4,
  number   = 4,
  pages    = {e00596},
  doi      = {https://doi.org/10.1016/j.heliyon.2018.e00596},
  issn     = {2405-8440},
  url      = {https://www.sciencedirect.com/science/article/pii/S2405844018301695}
}

@article{gen/rhg,
  title     = {Hyperbolic geometry of complex networks},
  author    = {Krioukov, Dmitri and Papadopoulos, Fragkiskos and Kitsak, Maksim and Vahdat, Amin and Bogu\~n\'a, Mari\'an},
  year      = 2010,
  month     = {Sep},
  journal   = {Phys. Rev. E},
  publisher = {American Physical Society},
  volume    = 82,
  pages     = {036106},
  doi       = {10.1103/PhysRevE.82.036106},
  url       = {https://link.aps.org/doi/10.1103/PhysRevE.82.036106},
  issue     = 3,
  numpages  = 18
}

@article{GeomBounds,
  title     = {{Geometric Bounds for Eigenvalues of Markov Chains}},
  author    = {Persi Diaconis and Daniel Stroock},
  year      = 1991,
  journal   = {The Annals of Applied Probability},
  publisher = {Institute of Mathematical Statistics},
  volume    = 1,
  number    = 1,
  pages     = {36 -- 61},
  doi       = {10.1214/aoap/1177005980},
  url       = {https://doi.org/10.1214/aoap/1177005980}
}

@article{gilbert1959random,
  title     = {Random Graphs},
  author    = {Gilbert, E. N.},
  year      = 1959,
  journal   = {The Annals of Mathematical Statistics},
  publisher = {Institute of Mathematical Statistics},
  volume    = 30,
  number    = 4,
  doi       = {10.1214/aoms/1177706098}
}

@book{gotelli1996null,
  title     = {Null models in ecology},
  author    = {Gotelli, N. J. and Graves, G. R.},
  year      = 1996,
  publisher = {Smithsonian Institution}
}

@article{KAUR201621,
  title    = {Analyzing negative ties in social networks: A survey},
  author   = {Mankirat Kaur and Sarbjeet Singh},
  year     = 2016,
  journal  = {Egyptian Informatics J.},
  volume   = 17,
  number   = 1,
  pages    = {21--43},
  doi      = {https://doi.org/10.1016/j.eij.2015.08.002},
  issn     = {1110-8665},
  url      = {https://www.sciencedirect.com/science/article/pii/S1110866515000432}
}

@book{levin2017markov,
  title     = {Markov chains and mixing times},
  author    = {Levin, David A and Peres, Yuval},
  year      = 2017,
  publisher = {American Mathematical Soc.},
  volume    = 107
}

@article{Milo824,
  title     = {Network Motifs: Simple Building Blocks of Complex Networks},
  author    = {Milo, R.},
  year      = 2002,
  journal   = {Science},
  publisher = {American Association for the Advancement of Science ({AAAS})},
  volume    = 298,
  number    = 5594,
  doi       = {10.1126/science.298.5594.824}
}

@book{Ford56,
  title     = {{Network Flow Theory}},
  author    = {Lester R. Ford Jr.},
  year      = 1956,
  publisher = {RAND Corporation},
  address   = {Santa Monica, CA}
}

@techreport{pcgrng,
  title       = {PCG: A Family of Simple Fast Space-Efficient Statistically Good Algorithms for Random Number Generation},
  author      = {Melissa E. O'Neill},
  year        = 2014,
  month       = Sep,
  address     = {Claremont, CA},
  number      = {HMC-CS-2014-0905},
  institution = {Harvey Mudd College},
  xurl        = {https://www.cs.hmc.edu/tr/hmc-cs-2014-0905.pdf}
}

@article{Peixto15,
  title     = {Model Selection and Hypothesis Testing for Large-Scale Network Models with Overlapping Groups},
  author    = {Peixoto, T. P.},
  year      = 2015,
  journal   = {Phys. Rev. X},
  publisher = {American Physical Society ({APS})},
  volume    = 5,
  number    = 1,
  doi       = {10.1103/physrevx.5.011033}
}

@article{rao1996markov,
  title     = {A Markov chain Monte Carlo method for generating random (0, 1)-matrices with given marginals},
  author    = {Rao, A. R. and Jana, R. and Bandyopadhyay, S.},
  year      = 1996,
  journal   = {Sankhy{\=a}: The Indian J. Statistics, Series A},
  publisher = {JSTOR}
}

@article{RobertsRosenthal04,
  title     = {{General state space Markov chains and MCMC algorithms}},
  author    = {Gareth O. Roberts and Jeffrey S. Rosenthal},
  year      = 2004,
  journal   = {Probability Surveys},
  publisher = {Institute of Mathematical Statistics and Bernoulli Society},
  volume    = 1,
  number    = {none},
  pages     = {20 -- 71},
  doi       = {10.1214/154957804100000024},
  url       = {https://doi.org/10.1214/154957804100000024}
}

@article{sinclair1992improved,
  title     = {Improved bounds for mixing rates of Markov chains and multicommodity flow},
  author    = {Sinclair, Alistair},
  year      = 1992,
  journal   = {Combinatorics, probability and Computing},
  publisher = {Cambridge University Press},
  volume    = 1,
  number    = 4,
  pages     = {351--370}
}

@article{strona2014curveball,
  title   = {A fast and unbiased procedure to randomize ecological binary matrices with fixed row and column totals},
  author  = {Strona, G. and Nappo, D. and Boccacci, F. and Fattorini, S. and Miguel-Ayanz, J.},
  year    = 2014,
  journal = {Nature Commun.},
  doi     = {https://doi.org/10.1038/ncomms5114}
}

@article{tinney1967direct,
  title     = {Direct solutions of sparse network equations by optimally ordered triangular factorization},
  author    = {Tinney, William F and Walker, John W},
  year      = 1967,
  journal   = {Proc. of the IEEE},
  publisher = {IEEE},
  volume    = 55,
  number    = 11,
  pages     = {1801--1809}
}

@inproceedings{UniformDist,
  title     = {Polynomial coefficients and distribution of the sum of discrete uniform variables},
  author    = {Caiado, Camila C. S. and Rathie, Pushpa N.},
  year      = 2007,
  booktitle = {{Proc. of SSFA}}
}

@article{10.1145/321694.321699,
  author     = {Edmonds, Jack and Karp, Richard M.},
  title      = {Theoretical Improvements in Algorithmic Efficiency for Network Flow Problems},
  year       = {1972},
  issue_date = {April 1972},
  publisher  = {Association for Computing Machinery},
  address    = {New York, NY, USA},
  volume     = {19},
  number     = {2},
  issn       = {0004-5411},
  url        = {https://doi.org/10.1145/321694.321699},
  doi        = {10.1145/321694.321699},
  journal    = {J. ACM},
  month      = {apr},
  pages      = {248–264},
  numpages   = {17}
}

@article{wasserstein,
 ISSN = {00251909, 15265501},
 URL = {http://www.jstor.org/stable/2627082},
 author = {L. V. Kantorovich},
 journal = {Management Science},
 number = {4},
 pages = {366--422},
 publisher = {INFORMS},
 title = {Mathematical Methods of Organizing and Planning Production},
 urldate = {2026-04-16},
 volume = {6},
 year = {1960}
}

\clearpage

\appendix

\section{Additional experimental results}
\begin{table}[h!]
  \caption{Number of possible~$|\wInt^n|$ and
  consistent~$|\setCons_{C_n, \wInt}|$ weights on the $n$-cycle.}
  \label{tab:ncycle_acceptance}

  \begin{center}
    \begin{tabular}{l|r|r|r}
      \multicolumn{1}{c|}{$n$} & \multicolumn{1}{c|}{$|\wInt^n|$} &
      \multicolumn{1}{c|}{$|\setCons_{C_n, \wInt}|$} &
      \multicolumn{1}{c}{Acc. rate} \\\hline
      8                        & \num{6561}                       &
      \num{3834}                                     & \num{0.584}
      \\
      12                       & \num{531441}                     &
      \num{302615}                                   & \num{0.569}
      \\
      16                       & \num{43046721}                   &
      \num{24121674}                                 & \num{0.560}
      \\
    \end{tabular}
  \end{center}
\end{table}

\begin{figure}[!htb]
  \centering
  \begin{subfigure}{0.5\textwidth}
    \centering
    \includegraphics[width=0.9\textwidth]{figures/plots/stat_tests/legend.pdf}
  \end{subfigure}
  \begin{subfigure}{0.9\textwidth}
    \centering
    \includegraphics[width=0.9\textwidth]{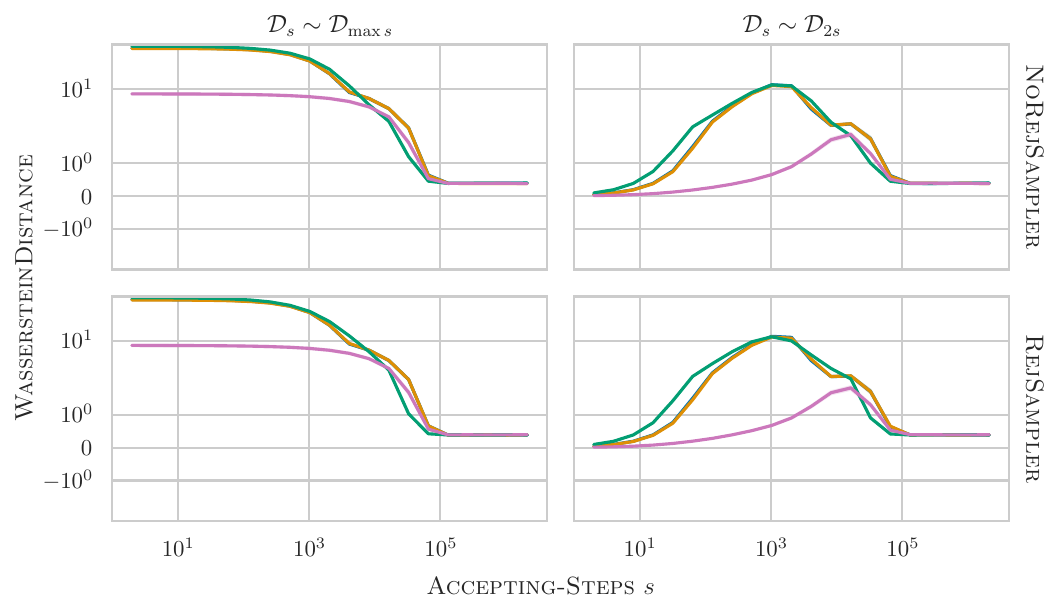}
  \end{subfigure}
  \begin{subfigure}{0.9\textwidth}
    \centering
    \includegraphics[width=0.9\textwidth]{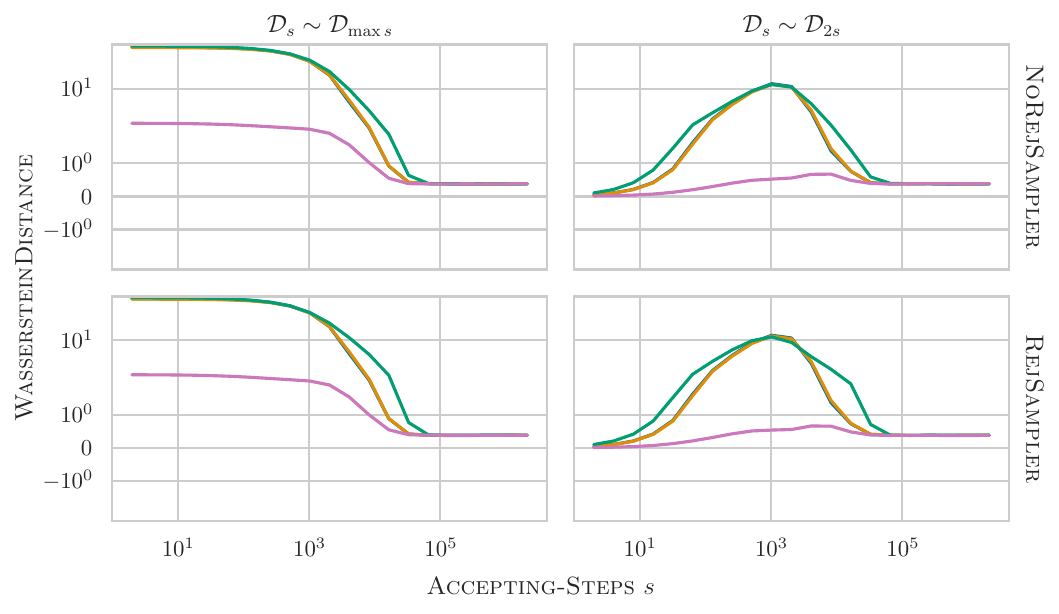}
  \end{subfigure}
  \caption{%
    Wasserstein-Distances for weight-distributions of $\gnp(n = 100,
    \avgdeg = 20)$ (up) and $\rhg(n = 100, \avgdeg = 20)$ (down) graphs.
    Weights are categorized by sampler and initial weight function;
    meta-distributions are then constructed for each edge (index),
    evaluated, and then aggregated in these plots.
    \vspace{-1em}
  }
  \label{fig:ws_stats_20}
\end{figure}

\def\mixtime{\ensuremath{\tau_{\epsilon}}\xspace}
\section{Rapid-Mixing on the $n$-Cycle}\label{sec:rapid-mixing}
We showed in \cref{thm:convergence-mc} that our algorithms converge
to a uniform distribution over $\wInt$ and saw in
\cref{sec:experiments} that in practice this happens in reasonable
polynomial time.
However, we still lack a theoretical upper bound on the mixing time.
As such proofs are generally hard and very involved, we restrict
ourselves to the $n$-Cycle $C_n$ and prove that our MCMC is rapidly
mixing on it for a set of integer weights $\wInt = [a,b]$.

Consider an ergodic Markov chain~$M$ with state space\footnote{
  Commonly, this is denoted by $\Omega$ in existing literature as
  opposed to $\setCons$.
} $\setCons$ and stationary distribution~$\pi$ (\ie the unique
configuration the ergodic~$M$ converges to).
The mixing time \mixtime of~$M$ describes the smallest number of
steps to reach a distribution that has at most $\epsilon$ distance
from $\pi$ (in $L_1$ norm) --- even for an adversarial starting configuration.
Formally, the mixing time is defined as
$$
\mixtime = \min \Big\{ \mcsteps \geq 0 : \max_{x \in \setCons} \Big[
    \sum_{y \in \setCons} |\pi(y) - \sigma_{\mcsteps,x}(y)| \leq \epsilon
\Big] \Big\},
$$
where $\sigma_{\mcsteps,x}$ is the distribution of states after
\mcsteps steps when starting on state $x$.

To show rapid mixing of the Markov chain underlying
\Cref{algo:sampler} for the $n$-cycle, we consider a variant that is
simpler to analyze.
This chain $\mathcal{M}$ is defined as follows.
Its state space $\setCons$ is the state space of the chain underlying
\Cref{algo:sampler}.
From the current state, the chain $\mathcal{M}$ transitions to the
next state as follows:
\begin{itemize}
  \item With probability $1 / 2$: remain in the current state.
  \item With probability $1 / 4$: choose edge $e$ and $b \in \{-1,
    1\}$ uniformly at random.
    If adding $b$ to the weight $w(e)$ yields a state in $\setCons$,
    move to this state.
  \item With probability $1 / 4$: choose edges $e_1$, $e_2$ uniformly at random.
    If incrementing $w(e_1)$ and decrementing $w(e_2)$ results in a
    state in $\setCons$, move to this state.
\end{itemize}
Since the same transitions are performed by the chain underlying
\Cref{algo:sampler} with a probability bounded from below by a
polynomial in $n$ and $(b - a)$, the mixing time of both chains can
be related by a polynomial in $n$ and $(b - a)$ (cf. \cite[Remark
13.19]{levin2017markov}).

The main result of this section is as follows.
\begin{theorem}
  The mixing time \mixtime of the Markov chain $\mathcal{M}$ for
  integer weights in $[a, b]$ on the $n$-cycle satisfies
  $
  \mixtime \leq 2^{9} n^{11} b^2 (b - a) \left(\log (b - a) + \log
  \epsilon^{-1} \right).
  $
  \label{thm:rapid-mixing}
\end{theorem}

\noindent
We employ a common argument to show Theorem \autoref{thm:rapid-mixing}.
Let $1 = \lambda_0 > \lambda_1 \geq \dots \geq \lambda_{|\setCons|} >
-1$ denote the eigenvalues of the transition matrix $P$ of $\mathcal{M}$.
Since $\mathcal{M}$ is lazy, e.g. we remain in the same state with
probability at least $1 / 2$, all eigenvalues are non-negative, and
it holds that $\mixtime \leq (1 - \lambda_1)^{-1} \left(\log
|\setCons| + \log \epsilon^{-1} \right)$, \eg~\cite[p. 4]{sinclair1992improved}.
Thus it suffices to bound $1 - \lambda_1$ from below.

Our idea is to decompose $\mathcal{M}$ into multiple Markov chains
$\mathcal{M'}$ and $\mathcal{M}_i$ where $i \in [0, nb]$.
Intuitively, $\mathcal{M'}$ moves between states that correspond to
the different sums of the weights in the $n$-cycle, and each
$\mathcal{M}_i$ moves between states that correspond to the same sum $i$.
Precisely, we partition $\setCons$ into the sets $\Omega_0, \dots
\setCons_{nb}$ where $\setCons_i$ contains all states with a sum of
weights~$i$ (n.b. since we consider a cycle graph, the sum of
consistent weights is always non-negative).
We then define the chain $\mathcal{M}_i$ with state space
$\setCons_i$ and transition matrix $P_i(x, y) = P(x, y)$ where $x
\neq y \in \setCons_i$ (with the total remaining probability, the
chain stays in the current state).
In addition, we define the chain $\mathcal{M'}$ with state space
$\setCons'~=~[0, nb]$ and transition matrix $P'(i, j)~=~\sum_{x \in
\setCons_i} \sum_{y \in \setCons_j} P(x, y)$ if $|i - j| = 1$ where
$i, j \in \setCons'$ and where again, with the total remaining
probability, the chain stays in the current state.

\begin{lemma}[Corollary 3.3 of~\cite{DBLP:journals/cpc/MartinR06}]
  \label{lem:mc-decomp}
  Let $\beta > 0$ and $\gamma > 0$ such that $P(x, y) \geq \beta$ for
  all $x, y \in \setCons$ where $P(x, y) > 0$, and
  $\pi(\delta_i(\setCons_j) \geq \gamma \pi(\setCons_i)$ for all $i
    \neq j$ where $P(i, j) > 0$ and $\delta_i(\setCons_j)$ is the
    subset of states in $\setCons_j$ where $P(i, j) > 0$.
    Then $\Gap(P) \geq \beta \gamma \Gap(P') \min_i \Gap(P_i)$.
  \end{lemma}

  \begin{lemma}[Proposition 6 of~\cite{GeomBounds}]
    \label{lem:conductance}
    For a MC with state space $\setCons$ and transition matrix $P$,
    define $h$ as $$
    h = \min_{\substack{X \subset \setCons \\ \pi(X) \leq
    {1}/{2}}}\frac{\sum_{x \in X}\sum_{y \in \setCons \setminus X}
    \pi(x) P(x,y)}{\pi(X)},
    $$ where $P(x, y)$ is the probability of moving from $x$ to $y$
    and $\pi(X) = \sum_{x \in X}\pi(x)$.
    Then, $1 - 2h \leq \lambda_1 \leq 1 - h^2$.
  \end{lemma}

  \begin{proof}[Proof of \cref{thm:rapid-mixing}]
    By \Cref{lem:mc-decomp}, the result follows by showing sufficient
    lower bounds for $\beta, \gamma, \Gap(P')$ and $\min_i \Gap(P_i)$.

    We start with $\min_i \Gap(P_i)$.
    A generalized variant of this kind of chain called the
    load-exchange Markov chain has been analyzed in
    \cite{DBLP:conf/stacs/AmanatidisK23}.
    In fact, by combining \cite[Theorem
    4.2]{DBLP:conf/stacs/AmanatidisK23} and
    \cite[Cor.~2.5]{DBLP:conf/stacs/AmanatidisK23}, it immediately
    follows that $\min_i \Gap(P_i) \geq 1 / (n^3 (b - a) 2)$.

    To bound $\Gap(P')$, we use \Cref{lem:conductance}.
    Consider any adversary set of states $X \subset \setCons'$ as
    required for \Cref{lem:conductance}.
    Observe that if $i \in X$ is a state such that $P'(i,j) > 0$ for
    some $j \notin X$, then $P'(i,j) \geq 1 / 8 n$.
    In addition, for any $i$, $\pi(i)$ is proportional to the number
    of weights on the $n$-cycle which sum to $i$.
    Precisely, this number is given by the polynomial coefficient
    $\binom{n}{i}_{b - a + 1}$, see for example~\cite{AdvComb, UniformDist}.
    For our purposes, it suffices to observe that $\pi(i)$ has a peak
    around some value $i$, and monotonically decreases in both directions.
    Now, let $i^*$ be the state in $X$ such that $\pi(i^*)$ is
    maximal, and assume that the only state $i$ such that $P'(i,j) >
    0$ for some $j \notin X$ is a state $i < i^*$ (the other cases
    all result in better bounds or are symmetric).
    Then it holds that
    $\pi(0) \leq \pi(1) \leq \dots \leq \pi(i) \leq \dots \leq \pi(i^*)$.
    Moreover, \Cref{lem:conductance} requires $\pi(0) + \pi(1) +
    \dots + \pi(i - 1) \geq 1 / 2$, and thus $\pi(i) \geq 1 / 2 nb$
    and $\pi(i) / \pi(X) \geq 1 / nb$.
    Combining this estimate with the bound on $P'(i, j)$ then gives
    $\Gap(P') = 1 - \lambda_1' \geq h^2 \geq 1 / 64 n^4 b^2$.

    Finally, observe that $\beta \geq 1 / 4 n^2$ and $\gamma \geq 1 /
    n$, which concludes the proof.
  \end{proof}

  \section{Convergence of the General \rej}\label{sec:general-convergence}
  The underlying Markov Chain of \cref{algo:sampler} is defined on
  the state space \setConsFull.
  In the proof of \cref{thm:convergence-mc}, we assumed that $\wInt$
  is countable and that we can assign a non-zero probability to every
  value in $\wInt$, \ie $\frac{1}{|\wInt|}$.
  However, if $\wInt$ is uncountable, \eg we have real-valued $\wInt
  = [a,b]$, the state space \setCons might also be uncountable and
  the proof of \cref{thm:convergence-mc} does not apply.
  Not only that but also the notions of aperiodicity and
  irreducibility do not apply to such a Markov Chain.
  For that, we need similar but more general definitions.
  We thus assume for the remainder of this section that $\wInt$
  consists of an arbitrary finite union of intervals.

  A Markov Chain on a measureable set $\mathcal{X}$ is defined by
  transition probabilities $P(x, dy)$ where for each $x \in
  \mathcal{X}$ and measureable subset $A \subseteq \mathcal{X}$, the
  value $P(x, A)$ is the probability to move from $x$ to somewhere in $A$.
  We again denote $i$ steps of the Markov Chain by $P^{(i)}$.
  Starting with $P^{(1)}(x, dy) = P(x, dy)$, higher-order transition
  probabilities are then given by \[
    P^{(n + 1)}(x, A) = \int_{\mathcal{X}}P^{(n)}(x, dy)P(y, A).
  \]

  \begin{definition}[$\phi$-irreducibility~\cite{RobertsRosenthal04}]
    A general Markov Chain is $\phi$-irreducible if there exists a
    non-zero $\sigma$-finite measure $\phi$ on $\mathcal{X}$ such
    that for all $A \subseteq \mathcal{X}$ with $\phi(A) > 0$, and
    for all $x \in \mathcal{X}$, there exists a positive integer $n =
    n(x, A)$ such that $P^{(n)}(x, A) > 0$.
  \end{definition}

  \begin{definition}[Aperiodicity\cite{RobertsRosenthal04}]
    A general Markov Chain with stationary distribution
    $\pi({\cdot})$ is aperiodic if there do not exist $d \geq 2$ and
    disjoint subsets $\mathcal{X}_1,\ldots,\mathcal{X}_d \subseteq
    \mathcal{X}$ with $P(x, \mathcal{X}_{i + 1} = 1$ for all $x \in
      \mathcal{X}_i$ ($i \in [d - 1]$), and $P(x, \mathcal{X}_1) = 1$
      for all $x \in \mathcal{X}_d$ such that $\pi(\mathcal{X}_1) > 0$.
    \end{definition}

    \begin{theorem}[Theorem 4 of~\cite{RobertsRosenthal04}]
      \label{thm:convergence-general-mc}
      If a Markov chain on a state space with countably generated
      $\sigma$-algebra is $\phi$-irreducible and aperiodic and has a
      stationary distribution $\pi(\cdot)$, then for $\pi$-a.e. $x
      \in \mathcal X$,
      \[ \lim_{n \to \infty} \Vert P^n(x, \cdot) - \pi(\cdot) \Vert = 0. \]
    \end{theorem}

    Verifying \cref{thm:convergence-general-mc} for our proposed MC
    yields a proof for the general version of~\cref{thm:convergence-mc}.

    \begin{theorem}[Generalization of~\cref{thm:convergence-mc}]
      Let $G = (V, E)$ be a directed graph and $\wInt$ a real weight interval.
      Then, the \rej process converges to a uniform density on
      $\setConsFull$.\label{thm:general-uniform}
    \end{theorem}

    \begin{proof}
      Subsets of $\mathbb R^m$ equipped with the standard Borel
      $\sigma$-algebra are countably generated by open balls with
      rational centers and rational radii~\cite{RobertsRosenthal04}.

      Furthermore, since the Markov chain is symmetric, detailed
      balance is fulfilled, \ie for probability densities $s, s'$ it holds
      \[ \pi(s)P(s, s') = \pi(s')P(s', s), \]
      for a stationary density $\pi$ and in particular for the uniform density.
      By symmetry it holds $\int_{x \in A}P(x, B)dx = \int_{y \in
      B}P(y, A)dy$ for measurable $A, B \subseteq \mathcal X$ which
      extends to a weighting of $\pi$ when chosen as the uniform density.
      Here, stationary follows directly from detailed balance, \ie
      integrating over $\mathcal X$ yields
      \begin{align*}
        \int_{\mathcal X} \pi(x)P(x, y)dx &= \int_{\mathcal X}
        \pi(y)P(y, x)dx \\
        &= \pi(y)\int_{\mathcal X} P(y, x)dx \\
        &= \pi(y).
      \end{align*}

      With the existence of $\pi$, it remains to verify
      $\phi$-irreducibility and aperiodicity.
      Let $x \in \setConsFull$ and $\pi(A) > 0$.
      We emulate the process of changing the weights to values
      exceeding some of the values of $A$ (with positive measure) to
      subsequently enable a move to $A$ with positive probability.
      Due to continuity there exists constant $\epsilon > 0$ where
      some point $y \in A$ emits $B(\epsilon, y) \subseteq A$.
      Now, with positive probability we successively change the
      coordinates of $x$ from smallest to largest to a value that is
      larger than any coordinate of any point in $B(\varepsilon, y)$
      (guaranteed to exist due to continuity) and returning to a
      value in $B(\epsilon, y)$ after a total of $2m$ steps, yielding
      \[ P^{(2m)}(x, A) \ge P^{(2m)}(x, B(\epsilon, y)) > 0, \]
      where $\pi$ is a probability measure and as such non-zero and
      $\sigma$-finite.

      Suppose $\mathcal X_1$ and $\mathcal X_2$ are disjoint subsets
      of $\mathcal X$ both of positive $\pi$ measure and
      wlog.~$\mathbf{0} \in \setConsFull \cap \mathcal X_1$, with
      $P(x, \mathcal X_2) = 1$ for all $x \in \mathcal X_1$.
      This essentially represents an emigration from $\mathcal X_1$
      to $\mathcal X_2$ with probability~$1$.
      Observe that in case $G$ contains at least one cycle, the
      probability to introduce a negative cycle is positive, hence
      $P(\mathbf{0}, \mathcal X_1) > 0$ leads to a contradiction as
      therefore $P(\mathbf{0}, \mathcal X_2) < 1$.

      Otherwise, since no cycles exist in $G$, the set $\setConsFull$
      coincides with $\wInt^m$ and the \rej process degenerates to
      choosing uniform values in $\wInt$ and assigning them to
      randomly sampled edges.
      Since $\pi(\mathcal X_1) > 0$ by definition we can find a
      subset $A \subset \mathcal X_1$ where for $x \in A$ it holds
      $P(x, \mathcal X_1) \ge P(x, A) > 0$.
    \end{proof}

    \end{document}